\documentclass[11pt]{article}

\usepackage[american]{babel}
\usepackage[T1]{fontenc}
\usepackage[utf8]{inputenc}
\usepackage[text={14.0cm,19.0cm}]{geometry}
\usepackage{lmodern}
\usepackage{microtype}

\usepackage{amsmath, amssymb, amsthm, mathtools, thmtools}
\usepackage{wrapfig}
\usepackage{nicefrac}

\usepackage[inline]{enumitem}
\setlist[enumerate]{label=(\roman*),itemsep=-1\parsep,topsep=0.5\parsep}

\usepackage{hyperref}

\declaretheorem{theorem}
\declaretheorem[sibling=theorem]{lemma}
\declaretheorem[sibling=theorem]{corollary}
\theoremstyle{definition}

\declaretheorem[sibling=theorem]{algorithm}
\theoremstyle{remark}
\declaretheorem[sibling=theorem]{remark}

\newcommand{\sgroup}[2]{G_{#1,#2}}
\newcommand{\paramcard}{\kappa}
\newcommand{\paramgap}{\gamma}
\newcommand{\paramthresh}{\theta}
\newcommand{\speeds}{V}
\newcommand{\jobsizes}{P}
\newcommand{\vjobsizes}{\mathcal{P}}

\newcommand{\Cmax}{C_{\max}}

\newcommand{\machs}{\mathcal{M}}
\newcommand{\jobs}{\mathcal{J}}
\newcommand{\poly}{\mathrm{poly}}
\newcommand{\Opt}{\mathrm{OPT}}
\newcommand{\Oh}{\mathcal{O}}
\newcommand{\Conf}{\mathcal{C}}

\newcommand{\rjobs}{J}
\newcommand{\size}{\mathrm{size}}
\newcommand{\makespan}{T}
\newcommand{\rmakespan}{\bar{T}}
\newcommand{\rrmakespan}{\breve{T}}
\newcommand{\bjobsizes}{B}
\newcommand{\sjobsizes}{S}
\newcommand{\source}{\alpha}
\newcommand{\sink}{\omega}

\newcommand{\supp}{\mathrm{supp}}
\newcommand{\MILP}{\textnormal{MILP}}

\newcommand{\eps}{\varepsilon}

\newcommand{\ZZ}{\mathbb{Z}}

\DeclarePairedDelimiter\floor{\lfloor}{\rfloor}
\DeclarePairedDelimiter\ceil{\lceil}{\rceil}
\DeclarePairedDelimiter\set{\lbrace}{\rbrace}
\DeclarePairedDelimiterX\sett[2]{\lbrace}{\rbrace}{ #1 \,\delimsize| \,\mathopen{} #2 }

\title{An EPTAS for Scheduling on Unrelated Machines of Few Different Types\footnote{This work was partially supported
by the German Research Foundation (DFG) project JA 612/16-1. The current article is an extended version of the conference article \cite{WADSversion}}}

\author{Klaus Jansen \qquad Marten Maack\\[11pt]
Department of Computer Science, University of Kiel, 24118 Kiel, Germany\\ 
\{kj, mmaa\}@informatik.uni-kiel.de}

\begin{document}

\maketitle

\begin{abstract}
In the classical problem of scheduling on unrelated parallel machines, a set of jobs has to be assigned to a set of machines.
The jobs have a processing time depending on the machine and the goal is to minimize the makespan, that is the maximum machine load.
It is well known that this problem is NP-hard and does not allow polynomial time approximation algorithms with approximation guarantees smaller than $1.5$ unless P$=$NP.
We consider the case that there are only a constant number $K$ of machine types.
Two machines have the same type if all jobs have the same processing time for them.
This variant of the problem is strongly NP-hard already for $K=1$. 
We present an efficient polynomial time approximation scheme (EPTAS) for the problem, that is, for any $\eps > 0$ an assignment with makespan of length at most $(1+\eps)$ times the optimum can be found in polynomial time in the input length and the exponent is independent of $1/\eps$.
In particular we achieve a running time of $2^{\Oh(K\log(K) \nicefrac{1}{\eps}\log^4 \nicefrac{1}{\eps})}+\poly(|I|)$, where $|I|$ denotes the input length.
Furthermore, we study three other problem variants and present an EPTAS for each of them:
The Santa Claus problem, where the minimum machine load has to be maximized; the case of scheduling on unrelated parallel machines with a constant number of uniform types, where machines of the same type behave like uniformly related machines; and the multidimensional vector scheduling variant of the problem where both the dimension and the number of machine types are constant.
For the Santa Claus problem we achieve the same running time.
The results are achieved, using mixed integer linear programming and rounding techniques.
\end{abstract}

\section{Introduction}

We consider the problem of scheduling jobs on unrelated parallel machines---or unrelated scheduling for short---in which a set $\jobs$ of $n$ jobs has to be assigned to a set $\machs$ of $m$ machines.
Each job $j$ has a processing time $p_{ij}$ for each machine $i$ and the goal is to find a schedule $\sigma:\jobs\rightarrow\machs$ minimizing the \emph{makespan} $C_{\max}(\sigma)=\max_{i\in\machs}\sum_{j\in\sigma^{-1}(i)}p_{ij}$, i.e. the maximum machine load.
The problem is one of the classical scheduling problems studied in approximation.
In $1990$ Lenstra, Shmoys and Tardos \cite{LST90} showed that there is no approximation algorithm with an approximation guarantee smaller than $1.5$, unless P$=$NP.
Moreover, they presented a $2$-approximation, and closing this gap is a rather famous open problem in scheduling theory and approximation (see e.g. \cite{WS11}).

In particular, we study the special case where there is only a constant number $K$ of \emph{machine types}.
Two machines $i$ and $i'$ have the same type, if $p_{ij}=p_{i'j}$ holds for each job $j$.
In many application scenarios this variant is plausible, e.g., when considering computers which typically only have a very limited number of different types of processing units. 
We denote the processing time of a job $j$ on a machine of type $t\in[K]$ by $p_{tj}$ and assume that the input consist of the corresponding $K\times n$ processing time matrix together with machine multiplicities $m_t$ for each type $t$, yielding $m=\sum_{t\in[K]}m_t$.
Note that the case $K=1$ is equivalent to the classical scheduling on identical machines.
We also study three other variants of the problem:

\subparagraph{Santa Claus Problem.}

We consider the reverse objective of maximizing the minimum machine load, i.e. $C_{\min}(\sigma)=\min_{i\in\machs}\sum_{j\in\sigma^{-1}(i)}p_{ij}$.
This problem is known as max-min fair allocation or the Santa Claus problem.
The intuition behind these names is that the jobs are interpreted as goods (e.g. presents), the machines as players (e.g. children), and the processing times as the values of the goods from the perspective of the different players.
Finding an assignment that maximizes the minimum machine load, means therefore finding an allocation of the goods that is in some sense fair (making the least happy kid as happy as possible).
We will refer to the problem as Santa Claus problem in the following, but otherwise will stick to the scheduling terminology.

\subparagraph{Uniform Types.}

Two machines $i$ and $i'$ have the same \emph{uniform} machine type, if there is a scaling factor $s$ such that $p_{ij} = sp_{i'j}$ for each job $j$.
While jobs behave on machines of the same type like they do on identical machines, they behave of machines of the same uniform type like they do on uniformly related machines. 
Hence, we may assume that the input consists of job sizes $p_{tj}$ depending on the job $j$ and the uniform type $t$, together with uniform machine types $t_i$ and machine speeds $s_i$, such that $p_{ij}=p_{t_ij}/s_i$.

\subparagraph{Vector Scheduling.}

In the $D$-dimensional vector scheduling variant of unrelated scheduling, a processing time vector $p_{ij} = (p^{(1)}_{ij},\dots, p^{(D)}_{ij})$ is given for each job $j$ and machine $i$ and the makespan of a schedule $\sigma$ is defined as the maximum load any machine receives in any dimension: 
\[\Cmax(\sigma)  =\max_{i\in\machs}\Big\|\sum_{j\in\sigma^{-1}(i)} p_{ij}\Big\|_{\infty} = \max_{i\in\machs,d\in[D]}\sum_{j\in\sigma^{-1}(i)} p^{(d)}_{ij}\]
Machine types are defined correspondingly.
We consider the case that both $K$ and $D$ are constant and like in the one dimensional case we may assume that the input consist of processing time vectors depending on types and jobs, together with machine multiplicities.


\paragraph{Basic Concepts.}

We study polynomial time approximation algorithms: 
Given an instance $I$ of an optimization problem, an $\alpha$-approximation $A$ for this problem produces a solution in time $\poly(|I|)$, where $|I|$ denotes the input length.
For the objective function value $A(I)$ of this solution it is guaranteed that $A(I)\leq \alpha\Opt(I)$, in the case of an minimization problem, or $A(I)\geq (1/\alpha)\Opt(I)$, in the case of an maximization problem, where $\Opt(I)$ is the value of an optimal solution.
We call $\alpha$ the \emph{approximation guarantee} or \emph{rate} of the algorithm. 
In some cases a polynomial time approximation scheme (PTAS) can be achieved, that is, an $(1+\eps)$-approximation for each $\eps>0$.
If for such a family of algorithms the running time can be bounded by $f(1/\eps)\poly(|I|)$ for some computable function $f$, the PTAS is called \emph{efficient} (EPTAS), and if the running time is polynomial in both $1/\eps$ and $|I|$ it is called \emph{fully polynomial} (FPTAS).

\paragraph{Related Work.}

It is well known that the unrelated scheduling problem admits an FPTAS in the case that the number of machines is considered constant \cite{HS76} and we already mentioned the seminal work by Lenstra et al. \cite{LST90}.
Furthermore, the problem of unrelated scheduling with a constant number of machine types is strongly NP-hard, because it is a generalization of the strongly NP-hard problem of scheduling on identical parallel machines.
Therefore an FPTAS can not be hoped for in this case.
However, Wiese, Bonifaci and Baruah showed that there is a PTAS \cite{wiese2013partitioned}, and Wiese and Bonifaci \cite{BW12} gave an extended analysis for the vector scheduling case where both the dimension $D$ and $K$ are constant.
The authors do not present a detailed analysis of the running time, however the procedures involve guessing steps with $(m+1)^{K\kappa}$ possibilities, where $\kappa = (D/\eps)^{\Oh((\nicefrac{1}{\eps}\log(\nicefrac{D}{\eps}))^D)}$.
Gehrke, Jansen, Kraft and Schikowski \cite{GJKS16} presented a PTAS with an improved running time of $\Oh(Kn)+m^{\Oh(\nicefrac{K}{\eps^2})}(\log(m)/\eps)^{\Oh(K^{2})}$ for the regular one dimensional case of unrelated scheduling with a constant number of machine types.
On the other hand, Chen, Jansen and Zhang \cite{chen2014optimality} showed that there is no PTAS for scheduling on identical machines with running time $2^{(\nicefrac{1}{\eps})^{1-\delta}}$ for any $\delta>0$, unless the exponential time hypothesis fails.
Furthermore, the case $K = 2$ has been studied:
Imreh \cite{Imr03} designed heuristic algorithms with rates $2+ (m_1-1)/m_2$ and $4 - 2/m_1$, and  Bleuse et al. \cite{BKMMT15} presented an algorithm with rate $4/3 + 3/m_2$ and, moreover, a (faster) $3/2$-approximation, for the case that for each job the processing time on the second machine type is at most the one on the first. 
Moreover, Raravi and N{\'e}lis \cite{RN12} designed a PTAS for the case with two machine types.

Interestingly, unrelated scheduling is in P, if both the number of machine types and the number of job types is bounded by a constant.
This is implied by a recent result due to Chen, Marx, Ye and Zhang \cite{chen2017parameterized} building upon a result by Goemans and Rothvoss \cite{GR14}.
Job types are defined analogously to machine types, i.e., two jobs $j,j'$ have the same type, if $p_{ij}=p_{ij'}$ for each machine $i$.
In this case the matrix $(p_{ij})$ has only a constant number of distinct rows and columns.
Note that both the number of machine types and uniform machine types bounds the rank of this matrix.
However the case of unrelated scheduling where the matrix $(p_{ij})$ has constant rank turns out to be much harder: 
Already for the case with rank $3$ the problem is APX-hard \cite{chen2017parameterized} and for rank $4$ an approximation algorithm with rate smaller than $3/2$ can be ruled out, unless P$=$NP \cite{chen2014improved}.
In a rather recent work, Knop and Koutecký \cite{knop2016scheduling} considered the number of machine types as a parameter from the perspective of fixed parameter tractability.
They showed that unrelated scheduling is fixed parameter tractable for the parameters $K$ and $\max p_{i,j}$, that is, there is an algorithm with running time $f(K,\max p_{i,j})\poly(|I|)$ for some computable function $f$ that solves the problem to optimality.
Chen et al. \cite{chen2017parameterized} extended this, showing that unrelated scheduling is fixed parameter tractable for the parameters $\max p_{i,j}$ and the rank of the processing time matrix.

For the case that the number of machines is constant, the Santa Claus problem behaves similar to the unrelated scheduling problem:
there is an FPTAS that is implied by a result due to Woeginger \cite{Woe00}.
In the general case however, so far no approximation algorithm with a constant approximation guarantee has been found.
The results by Lenstra et al. \cite{LST90} can be adapted to show that that there is no approximation algorithm with a rate smaller than $2$, unless P$=$NP, and to get an algorithm that finds a solution with value at least $\Opt(I)-\max p_{i,j}$, as was done by Bezáková and Dani \cite{BD05}.
Since $\max p_{i,j}$ could be bigger than $\Opt(I)$, this does not provide a (multiplicative) approximation guarantee.
Bezáková and Dani also presented a simple $(n-m+1)$-approximation and an improved approximation guarantee of $\Oh(\sqrt{n}\log^3 n)$ was achieved by Asadpour and Saberi \cite{AS10}.
The best rate so far is $O(n^\eps)$ due to Bateni et al. \cite{BCG09} and Chakrabarty et al. \cite{CCK09}, with a running time of $\Oh(n^{1/\eps})$ for any $\eps>0$.

To the best of our knowledge, unrelated scheduling with a constant number of uniform machine types has not been studied before, but we argue that it is a natural extension of the case with a constant number of regular machine types and also a sensible special case of the general unrelated scheduling and the low rank case in particular.

The vector scheduling problem has been studied for the special case of identical machines by Chekuri and Khanna \cite{chekuri2004multidimensional}.
They achieve a PTAS for the case that $D$ is constant and an $\Oh(\log^2 D)$-approximation for the case that $D$ is arbitrary.

\paragraph{Results and Methodology.}

The main result of this paper is the following:
\begin{theorem}\label{thm:main_result}
There is an EPTAS for both scheduling on unrelated parallel machines and the Santa Claus problem with a constant number  $K$ of different machine types with running time $2^{\Oh(K\log(K) \nicefrac{1}{\eps}\log^4 \nicefrac{1}{\eps})}+\poly(|I|)$.
\end{theorem}
First we present a basic version of the EPTAS for unrelated scheduling with a running time doubly exponential in $1/\eps$.
For this EPTAS we use the dual approximation approach by Hochbaum and Shmoys \cite{HS87Dual} to get a guess $\makespan$ of the optimal makespan $\Opt$.
Then, we further simplify the problem via geometric rounding of the processing times.
Next, we formulate a mixed integer linear program (MILP) with a constant number of integral variables that encodes a relaxed version of the problem.
The MILP can be seen as a generalization of the classical integer linear program of configurations---or configuration ILP---for scheduling on identical parallel machines.
We solve it with the algorithm by Lenstra and Kannan \cite{Len83ILP,Kan87ILP}.  
The fractional variables of the MILP have to be rounded and we achieve this with a flow network utilizing flow integrality and causing only a small error. 
With an additional error the obtained solution can be used to construct a schedule with makespan $(1+\Oh(\eps))\makespan$.
This procedure is described in detail in Section \ref{sec:basic_eptas}.
Building upon the basic EPTAS we achieve the improved running time using techniques by Jansen \cite{Jan10QCmaxEPTAS} and by Jansen, Klein and Verschae \cite{JKV16ICALP}.
The basic idea of these techniques is to make use of existential results about simple structured solutions of integer linear programs (ILPs).
In particular these results can be used to guess the non-zero variables of the MILP, because they sufficiently limit the search space.
We show how these techniques can be applied in our case in Section \ref{sec:better_running_time}.
Furthermore, we present efficient approximation schemes for several other problem variants, thereby demonstrating the flexibility of our approach.
In particular, we can adapt all our techniques to the Santa Claus problem yielding the result stated above.
This is covered in Section \ref{sec:santa} and in Section \ref{sec:uniform} we show:
\begin{theorem}\label{thm:uniform}
There is an EPTAS for scheduling on unrelated parallel machines with a constant number $K$ of different uniform machine types with running time $2^{\Oh(K\log(K)\nicefrac{1}{\eps^3} \log^5\nicefrac{1}{\eps})}+\poly(|I|)$.
\end{theorem}
We achieve this with a non-trivial combination of the ideas of Section \ref{sec:basic_eptas} with techniques for scheduling on uniformly related machines by Jansen \cite{Jan10QCmaxEPTAS}.
Finally, in Section \ref{sec:vector}, we revisit the unrelated vector scheduling problem that was studied by Bonifaci and Wiese \cite{BW12}.
We show that an additional rounding step---similar to the one in \cite{chekuri2004multidimensional}---together with a slight modification of the MILP and the rounding procedure yield an EPTAS for this problem as well.
\begin{theorem}\label{thm:vector}
There is an EPTAS for vector scheduling on unrelated parallel machines with constant dimension $D$ a constant number $K$ of different machine types.
\end{theorem}
Note that our results may also be seen as fixed parameter tractable algorithms for the parameters $1/\eps$ and $K$ (and $D$).
In the last section we elaborate on possible directions for future research.

\section{Basic EPTAS}\label{sec:basic_eptas}

In this chapter we describe a basic EPTAS for unrelated scheduling with a constant number of machine types, with a running time doubly exponential in $1/\eps$.
Wlog. we assume $\eps< 1$.
Furthermore $\log(\cdot)$ denotes the logarithm with basis $2$ and for $k\in\ZZ_{\geq 0}$ we write $[k]$ for $\set{1,\dots,k}$.

First, we simplify the problem via the classical dual approximation concept by Hochbaum and Shmoys \cite{HS87Dual}.
In the simplified version of the problem a target makespan $\makespan$ is given and the goal is to either output a schedule with makespan at most $(1+\alpha \eps)\makespan$ for some constant $\alpha\in\ZZ_{>0}$, or correctly report that there is no schedule with makespan $\makespan$.
We can use a polynomial time algorithm for this problem in the design of a PTAS in the following way.
First we obtain an upper bound $B$ for the optimal makespan $\Opt$ of the instance with $B\leq 2\Opt$.
This can be done using the $2$-approximation by Lenstra et al. \cite{LST90}.
With binary search on the interval $[B/2,B]$ we can find in $\Oh(\log 1/\eps)$ iterations a value $\makespan^*$ for which the mentioned algorithm is successful, while $T^*-\eps B/2$ is rejected.
We have $T^*-\eps B/2\leq\Opt$ and therefore $T^*\leq (1+\eps)\Opt$.
Hence the schedule we obtained for the target makespan $T^*$ has makespan at most $(1+\alpha\eps)T^*\leq (1+\alpha\eps)(1+\eps)\Opt=(1+\Oh(\eps))\Opt$.
In the following we will always assume that a target makespan $\makespan$ is given.
Next we present a brief overview of the algorithm for the simplified problem followed by a more detailed description and analysis.

\begin{algorithm}
\ 
\begin{enumerate}
\item Simplify the input via geometric rounding with an error of $\eps\makespan$.
\item Build the mixed integer linear program $\MILP(\rmakespan)$ and solve it with the algorithm by Lenstra and Kannan ($\rmakespan=(1+\eps)\makespan$).
\item If there is no solution, report that there is no solution with makespan $T$.
\item Generate an integral solution for $\MILP(\rmakespan+\eps\makespan+\eps^2\makespan)$ via a flow network utilizing flow integrality.
\item The integral solution is turned into a schedule with an additional error of $\eps^2\makespan$ due to the small jobs.   
\end{enumerate}
\end{algorithm}

\paragraph{Simplification of the Input.}

We construct a simplified instance $\bar{I}$ with modified processing times $\bar{p}_{tj}$.
If a job $j$ has a processing time bigger than $T$ for a machine type $t\in[K]$ we set $\bar{p}_{tj}=\infty$.
We call a job \emph{big} (for machine type $t$), if $p_{tj}>\eps^2\makespan$, and \emph{small} otherwise.
We perform a geometric rounding step for each job $j$ with $p_{tj}<\infty$, that is we set $\bar{p}_{tj}=(1+\eps)^x\eps^{2} \makespan$ with $x=\ceil{\log_{1+\eps}(p_{tj}/(\eps^{2} \makespan))}$.
\begin{lemma}\label{lem:rounded_instance}
If there is a schedule with makespan at most $\makespan$ for $I$, the same schedule has makespan at most $(1+\eps)\makespan$ for instance $\bar{I}$ and any schedule for instance $\bar{I}$ can be turned into a schedule for $I$ without increase in the makespan.
\end{lemma}
We will search for a schedule with makespan $\rmakespan=(1+\eps)\makespan$ for the rounded instance $\bar{I}$.
We establish some notation for the rounded instance.
For any rounded processing time $p$ we denote the set of jobs $j$ with $\bar{p}_{tj}=p$ by $\rjobs_{t}(p)$.
Moreover, for each machine type $t$ let $\sjobsizes_t$ and $\bjobsizes_t$ be the sets of small and big rounded processing times.
Obviously we have $|\sjobsizes_t|+|\bjobsizes_t|\leq n$.
Furthermore $|\bjobsizes_t|$ is bounded by a constant:
Let $N$ be such that $(1+\eps)^N\eps^2\makespan$ is the biggest rounded processing time for all machine type. 
Then we have  $(1+\eps)^{N-1}\eps^{2} \makespan \leq \makespan$ and therefore $|\bjobsizes_t|\leq N\leq\log(1/\eps^{2})/\log(1+\eps) +1\leq 1/\eps\log(1/\eps^{2})+1$ (using $\eps\leq 1$).

\paragraph{MILP.}

For any set of processing times $P$ we call the $P$-indexed vectors of non-negative integers $\ZZ_{\geq 0}^P$ \emph{configurations} (for $P$).
The \emph{size} $\size(C)$ of configuration $C$ is given by $\sum_{p\in P}C_p p$.
For each $t\in[K]$ we consider the set $\Conf_t(\rmakespan)$ of configurations $C$ for the big processing times $\bjobsizes_t$ and with $\size(C)\leq\rmakespan$.
Given a schedule $\sigma$, we say that a machine $i$ of type $t$ obeys a configuration $C$, if the number of big jobs with processing time $p$ that $\sigma$ assigns to $i$ is exactly $C_p$ for each $p\in\bjobsizes_t$.
Since the processing times in $\bjobsizes_t$ are bigger than $\eps^{2}\makespan$ we have $\sum_{p\in\bjobsizes_t}C_p\leq 1/\eps^{2}$ for each $C\in\Conf_t(\rmakespan)$.
Therefore the number of distinct configurations in $\Conf_t(\rmakespan)$ can be bounded by $(1/\eps^{2}+1)^{N}< (1/\eps^{2}+1)^{\nicefrac{1}{\eps}\log(\nicefrac{1}{\eps^{2}})+1}=2^{\log(\nicefrac{1}{\eps^{2}}+1)\nicefrac{1}{\eps}\log(\nicefrac{1}{\eps^{2}})+1}\in 2^{\Oh(\nicefrac{1}{\eps}\log^2 \nicefrac{1}{\eps})}$.

We define a mixed integer linear program $\MILP(\rmakespan)$ in which configurations are chosen integrally and jobs are assigned fractionally to machine types.
Note that we will call a solution of a MILP integral if both the integral and fractional variables have integral values.
We introduce variables $z_{C,t}\in\ZZ_{\geq 0}$ for each machine type $t\in[K]$ and configuration $C\in\Conf_t(\rmakespan)$, and $x_{j,t}\geq 0$ for each machine type $t\in[K]$ and job $j\in\jobs$.
For $\bar{p}_{tj}=\infty$ we set $x_{j,t}=0$. 
Besides this, the MILP has the following constraints:
\begin{align}
\sum_{C\in\Conf_t(\rmakespan)}z_{C,t}                   &= m_t                               & \forall t\in [K] \label{eq:MILP_Conf} \\
\sum_{t\in[K]} x_{j,t}                     &= 1                                    & \forall j\in\jobs \label{eq:MILP_job_assignment} \\
\sum_{j\in\rjobs_{t}(p)}x_{j,t}            &\leq \sum_{C\in\Conf_t(\rmakespan)}C_{p}z_{C,t}   & \forall t\in [K],p\in \bjobsizes_t \label{eq:MILP_big_jobs_bounded_by_conf}\\
\sum_{C\in\Conf_t(\rmakespan)} \size(C) z_{C,t} + \sum_{p\in\sjobsizes_t}p\sum_{j\in\rjobs_{t}(p)} x_{j,t}   &\leq m_t\rmakespan                             & \forall t\in [K] \label{eq:MILP_jobs_fit_into_space}
\end{align}
With constraint (\ref{eq:MILP_Conf}) the number of chosen configurations for each machine type equals the number of machines of this type.
Due to constraint (\ref{eq:MILP_job_assignment}) the variables $x_{j,t}$ encode the fractional assignment of jobs to machine types.
Moreover for each machine type it is ensured with constraint (\ref{eq:MILP_big_jobs_bounded_by_conf}) that the summed up number of big jobs of each size is at most the number of big jobs that are used in the chosen configurations for the respective machine type.
Lastly, (\ref{eq:MILP_jobs_fit_into_space}) guarantees that the overall processing time of the configurations and small jobs assigned to a machine type does not exceed the area $m_t\rmakespan$.
It is easy to see that the MILP models a relaxed version of the problem:
\begin{lemma}\label{lem:MILP_and_schedule}
If there is schedule with makespan $\rmakespan$ there is a feasible (integral) solution of $\MILP(\rmakespan)$, and if there is a feasible integral solution for $\MILP(\rmakespan)$ there is a schedule with makespan at most $\rmakespan+\eps^2\makespan$.
\end{lemma}
\begin{proof}
Let $\sigma$ be a schedule with makespan $\rmakespan$.
Each machine of type $t$ obeys exactly one configuration from $\Conf_t(\rmakespan)$, and we set $z_{C,t}$ to be the number of machines of type $t$ that obey $C$ with respect to $\sigma$.
Furthermore for a job $j^*$ let $t^*$ be the type of machine $\sigma(j^*)$.
We set $x_{j^*,t^*}=1$ and $x_{j^*,t}=0$ for $t\neq t^*$.
It is easy to check that all conditions are fulfilled.

Now let $(z_{C,t},x_{j,t})$ be an integral solution of $\MILP(\rmakespan)$.
Using (\ref{eq:MILP_job_assignment}) we can assign the jobs to distinct machine types based on the $x_{j,t}$ variables.
The $z_{C,t}$ variables can be used to assign configurations to machines such that each machine receives exactly one configuration using (\ref{eq:MILP_Conf}).
Based on these configurations we can create slots for the big jobs and for each type $t$ we can successively assign all of the big jobs assigned to this type to slots of the size of their processing time, because of (\ref{eq:MILP_big_jobs_bounded_by_conf}).
Now, for each type, we can iterate through the machines and greedily assign small jobs.
When the makespan $\rmakespan$ is exceeded due to some job, we stop assigning to the current machine and continue with the next.
Because of (\ref{eq:MILP_jobs_fit_into_space}), all small jobs can be assigned in this fashion.
Since the small jobs have size at most $\eps^{2}\makespan$, we get a schedule with makespan at most $\rmakespan+\eps^2\makespan$.
\end{proof}

We have $K2^{\Oh(\nicefrac{1}{\eps}\log^2 \nicefrac{1}{\eps})}$ integral variables, i.e., a constant number.
Therefore $\MILP(\makespan)$ can be solved in polynomial time, with the following classical result due to Lenstra \cite{Len83ILP} and Kannan \cite{Kan87ILP}:
\begin{theorem}\label{thm:lenstra_kannan_ilp}
A mixed integer linear program with $d$ integral variables and encoding size $s$ can be solved in time $d^{\Oh(d)}\poly(s)$. 
\end{theorem}

\paragraph{Rounding.}

In this paragraph we describe how a feasible solution $(z_{C,t},x_{j,t})$ for $\MILP(\rmakespan)$ can be transformed into an integral feasible solution $(\bar{z}_{C,t},\bar{x}_{j,t})$ for $\MILP(\rmakespan+\eps\makespan+\eps^2\makespan)$, where the second MILP is defined using the same configurations but accordingly changed right hand side.
This is achieved via a flow network utilizing flow integrality.

For any (small or big) processing time $p$ let $\eta_{t,p}=\ceil{\sum_{j\in\rjobs_{t}(p)}x_{j,t}}$ be the rounded up (fractional) number of jobs with processing time $p$ that are assigned to machine type $t$.
Note that for big job sizes $p\in\bjobsizes_t$, we have $\eta_{t,p}\leq\sum_{C\in\Conf_t(\rmakespan)}C_{p}z_{C,t}$, because of (\ref{eq:MILP_big_jobs_bounded_by_conf}) and because the right hand side is an integer.

\begin{figure}
\centering
\includegraphics{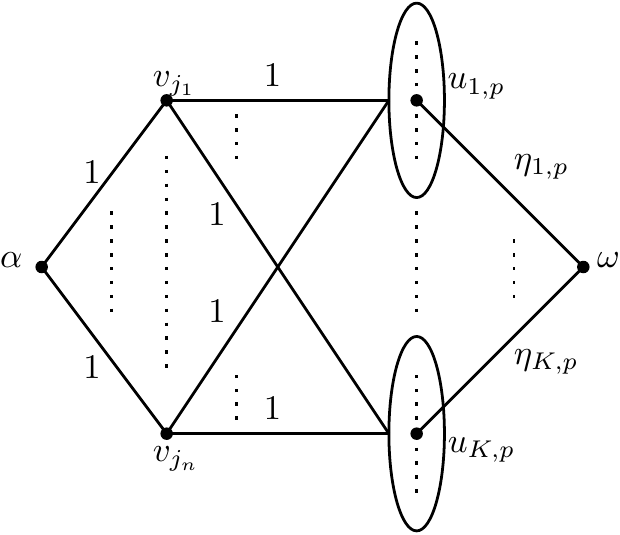}
\caption{A sketch of the flow network.}
\label{fig:flow}
\end{figure}
Now we describe the flow network $G=(V,E)$ with source $\source$ and sink $\sink$.
For each job $j\in\jobs$ there is a job node $v_j$ and an edge $(\source,v_j)$ with capacity $1$ connecting the source and the job node.
Moreover, for each machine type $t$ we have processing time nodes $u_{t,p}$ for each processing time $p\in\bjobsizes_t\cup\sjobsizes_t$.
The processing time nodes are connected to the sink via edges $(u_{t,p},\sink)$ with capacity $ \eta_{t,p}$.
Lastly, for each job $j$ and machine type $t$ with $\bar{p}_{t,j}<\infty$, we have an edge $(v_j,u_{t,\bar{p}_{t,j}})$ with capacity $1$ connecting the job node with the corresponding processing time nodes.
We outline the construction in Figure \ref{fig:flow}.
Obviously we have $|V|\leq (K+1)n+2$ and $|E|\leq (2K+1)n$.
\begin{lemma}\label{lem:max_flow}
$G$ has a maximum flow with value $n$.
\end{lemma}
\begin{proof}
Since the outgoing edges from $\source$ have summed up capacity $n$, $n$ is a trivial upper bound for the maximum flow.
The solution $(z_{C,t},x_{j,t})$ for $\MILP(\rmakespan)$ can be used to design a flow $f$ with value $n$, by setting $f((\source,v_j))=1$, $f((v_j,u_{t,\bar{p}_{t,j}}))=x_{j,t}$ and $f((u_{t,y},\sink))=\sum_{j\in\rjobs_{t}(y)}x_{j,t}$.
It is easy to check that $f$ is indeed a feasible flow with value $n$.
\end{proof}
Using the Ford-Fulkerson algorithm, an integral maximum flow $f^*$ can be found in time $\Oh(|E|f^*)=\Oh(Kn^2)$.
Due to flow conservation, for each job $j$ there is exactly one machine type $t^*$ such that $f((v_j,u_{t^*,\bar{p}_{t^*,j}}))=1$, and we set $\bar{x}_{j,t^*}=1$ and $\bar{x}_{j,t}=0$ for $t\neq t^*$.
Moreover, we set $\bar{z}_{C,t}=z_{C,t}$.
Obviously $(\bar{z}_{C,t},\bar{x}_{j,t})$ fulfils (\ref{eq:MILP_Conf}) and (\ref{eq:MILP_job_assignment}).
Furthermore, (\ref{eq:MILP_big_jobs_bounded_by_conf}) is fulfilled, because of the capacities and because $\eta_{t,p}\leq\sum_{C\in\Conf_t(\rmakespan)}C_{p}z_{C,t}$ for big job sizes $p$.
Utilizing the geometric rounding and the convergence of the geometric series, as well as $\sum_{j\in\rjobs_{t}(p)} \bar{x}_{j,t}\leq\eta_{t,p}<\sum_{j\in\rjobs_{t}(p)} x_{j,t}+1$, we get: 
\begin{equation*}
\sum_{p\in\sjobsizes_t}p\sum_{j\in\rjobs_{t}(p)} \bar{x}_{j,t} < \sum_{p\in\sjobsizes_t}p\sum_{j\in\rjobs_{t}(p)} x_{j,t} + \sum_{p\in\sjobsizes_t}p < \sum_{p\in\sjobsizes_t}p\sum_{j\in\rjobs_{t}(p)} x_{j,t} + \eps^2\makespan\frac{1+\eps}{\eps}
\end{equation*}
Hence, we have $\sum_{C\in\Conf_t(\rmakespan)} \size(C) \bar{z}_{C,t} + \sum_{p\in\sjobsizes_t}p\sum_{j\in\rjobs_{t,s}} \bar{x}_{j,t} < m_t(\rmakespan +\eps T + \eps^2 T)$ and therefore (\ref{eq:MILP_jobs_fit_into_space}) is fulfilled as well.

\paragraph{Analysis.}

The solution found for $\MILP(\rmakespan)$ can be turned into an integral solution for $\MILP(\rmakespan+\eps\makespan+\eps^2\makespan)$.
Like described in the proof of Lemma \ref{lem:MILP_and_schedule} this can easily be turned into a schedule with makespan $\rmakespan+\eps\makespan+\eps^2\makespan+\eps^2\makespan\leq(1+4\eps)\makespan$.
It is easy to see that the running time of the algorithm by Lenstra and Kannan dominates the overall running time.
Since $\MILP(\rmakespan)$ has $\Oh(K/\eps\log 1/\eps+n)$ many constraints, $Kn$ fractional and $K2^{\Oh(\nicefrac{1}{\eps}\log^2 \nicefrac{1}{\eps})}$ integral variables, the running time of the algorithm can be bounded by:
\[(K2^{\Oh(\nicefrac{1}{\eps}\log^2 \nicefrac{1}{\eps})})^{\Oh(K2^{\Oh(\nicefrac{1}{\eps}\log^2\nicefrac{1}{\eps})})}\poly((K/\eps\log 1/\eps)|I|) = 2^{K2^{\Oh(\nicefrac{1}{\eps}\log^2 \nicefrac{1}{\eps})}}\poly(|I|)\]

\section{Better running time}\label{sec:better_running_time}

We improve the running time of the algorithm using techniques that utilize results concerning the existence of solutions for integer linear programs (ILPs) with a certain simple structure.
In a first step we can reduce the running time to be only singly exponential in $1/\eps$ with a technique by Jansen \cite{Jan10QCmaxEPTAS}.
Then we further improve the running time to the one claimed in Theorem \ref{thm:main_result} with a very recent result by Jansen, Klein and Verschae \cite{JKV16ICALP}.
Both techniques rely upon the following result about integer cones by Eisenbrandt and Shmonin \cite{ES06caratheodory}.
\begin{theorem}\label{thm:Eisenbrandt_Shmonin}
Let $X\subset\ZZ^d$ be a finite set of integer vectors and let $b\in\textnormal{int-cone}(X)=\sett{\sum_{x\in X}\lambda_x x}{\lambda_x\in\ZZ_{\geq 0}}$.
Then there is a subset $\tilde{X}\subseteq X$, such that $b\in\textnormal{int-cone}(\tilde{X})$ and $|\tilde{X}|\leq 2d\log(4dM)$, with $M=\max_{x\in X}\| x\|_\infty$.
\end{theorem}
For the first improvement of the running time, this theorem is used to show:
\begin{corollary}\label{cor:few_variables}
$\MILP(\rmakespan)$ has a feasible solution, where for each machine type at most $\Oh(1/\eps \log^2 1/\eps)$ of the corresponding integer variables are non-zero.
\end{corollary}
We get the better running time by guessing the non-zero variables and removing all the others from the MILP.
The number of possibilities of choosing $\Oh(1/\eps \log^2 1/\eps)$ elements out of a set of $2^{\Oh(\nicefrac{1}{\eps}\log^2 \nicefrac{1}{\eps})}$ elements can be bounded by $2^{\Oh(\nicefrac{1}{\eps^2}\log^4 \nicefrac{1}{\eps})}$.
Considering all the machine types we can bound the number of guesses by $2^{\Oh(\nicefrac{K}{\eps^2}\log^4 \nicefrac{1}{\eps})}$.
The running time of the algorithm by Lenstra and Kannan with $\Oh(K/\eps \log^2 1/\eps)$ integer variables can be bounded by: 
\[\Oh(K/\eps \log^2 1/\eps)^{\Oh(\nicefrac{K}{\eps}\log^2 \nicefrac{1}{\eps})}\poly(|I|)=2^{\Oh(K\log(K)\nicefrac{1}{\eps} \log^3\nicefrac{1}{\eps})}\poly(|I|)\]
This yields a running time of: 
\[2^{\Oh(K\log(K)\nicefrac{1}{\eps^2} \log^4 \nicefrac{1}{\eps})}\poly(|I|)\]

In the following we first proof Corollary \ref{cor:few_variables} and then introduce the technique from \cite{JKV16ICALP} to further reduce the running time.

\paragraph{Proof of Corollary \ref{cor:few_variables}.}

We consider the so called \emph{configuration ILP} for scheduling on identical machines.
Let $m'$ be a given number of machines, $P$ be a set of processing times with multiplicities $k_p\in\ZZ_{>0}$ for each $p\in P$ and let $\Conf\subseteq \ZZ_{\geq 0}^P$ be some finite set of configurations for $P$.
The configuration ILP for $m'$, $P$, $k=(k_p)_{p\in P}$, and $\Conf$ is given by:
\begin{align}
\sum_{C\in\Conf} C_p y_{C} & = k_p & \forall p\in P \\ 
\sum_{C\in\Conf}y_{C} & = m' &\\
y_C & \in \ZZ_{\geq 0}   & \forall C\in\Conf
\end{align}
The default case that we will consider most of the time is that $\Conf$ is given by a target makespan $\makespan$ that upper bounds the size of the configurations. 

Let's assume we had a feasible solution $(\tilde{z}_{C,t},\tilde{x}_{j,t})$ for $\MILP(\rmakespan)$.
For $t\in[K]$ and $p\in\bjobsizes_t$ we set $\tilde{k}_{t,p}=\sum_{C\in\Conf_t(\rmakespan)}C_p \tilde{z}_{C,t}$.
We fix a machine type $t$.
By setting $y_{C}=\tilde{z}_{C,t}$, we get a feasible solution for the configuration ILP given by $m_t$, $\bjobsizes_t$, $\tilde{k}_{t}$ and $\Conf_t(\rmakespan)$.
Theorem \ref{thm:Eisenbrandt_Shmonin} can be used to show the existence of a solution for the ILP with only a few non-zero variables:
Let $X$ be the set of column vectors corresponding to the left hand side of the ILP and $b$ be the vector corresponding to the right hand side.
Then $b\in\textnormal{int-cone}(X)$ holds and Theorem \ref{thm:Eisenbrandt_Shmonin} yields that there is a subset $\tilde{X}$ of $X$ with cardinality at most $2(|\bjobsizes_t|+1)\log(4(|\bjobsizes_t|+1)1/\eps^2)\in\Oh(1/\eps \log^2 1/\eps)$ and $b\in\textnormal{int-cone}(\tilde{X})$.
Therefore there is a solution $(\breve{y}_C)$ for the ILP with $\Oh(1/\eps \log^2 1/\eps)$ many non-zero variables.
If we set $\breve{z}_{C,t}=\breve{y}_C$ and $\breve{x}_{j,t}=\tilde{x}_{j,t}$ and perform corresponding steps for each machine type, we get a solution $(\breve{z}_{C,t},\breve{x}_{j,t})$ that obviously satisfies constraints (\ref{eq:MILP_Conf}),(\ref{eq:MILP_job_assignment}) and (\ref{eq:MILP_big_jobs_bounded_by_conf}) of $\MILP(\rmakespan)$.
The last constraint is also satisfied, because the number of covered big jobs of each size does not change and therefore the overall size of the configurations does not change either for each machine type.
This completes the proof of Corollary \ref{cor:few_variables}.

\paragraph{Further Improvement of the Running Time.}

The main ingredient of the technique by Jansen et al. \cite{JKV16ICALP} is a result about the configuration ILP, for the case that there is a target makespan $\makespan'$ upper bounding the configuration sizes.
Let $\Conf(\makespan')$ be the set of configurations with size at most $\makespan'$.
We need some further notation.
The \emph{support} of any vector of numbers $v$ is the set of indices with non-zero entries, i.e., $\supp(v)=\sett{i}{v_i\neq 0}$.
A configuration is called \emph{simple}, if the size of its support is at most $\log(\makespan'+1)$, and \emph{complex} otherwise.
The set of complex configurations from $\Conf(\makespan')$ is denoted by $\Conf^c(\makespan')$.
\begin{theorem}\label{thm:thin_solutions}
Let the configuration ILP for $m'$, $P$, $k$, and $\Conf(\makespan')$ have a feasible solution and let both the makespan $\makespan'$ and the processing times from $P$ be integral.
Then there is a solution $(y_C)$ for the ILP that satisfies the following conditions:
\begin{enumerate}
\item $|\supp(y|_{\Conf^c(\makespan')})|\leq 2(|P|+1)\log(4(|P|+1)\makespan')$ and $y_C\leq 1$ for $C\in\Conf^c(\makespan')$.
\item $|\supp(y)|\leq 4(|P|+1)\log(4(|P|+1)\makespan')$.
\end{enumerate}
\
\end{theorem}
We will call such a solution \emph{thin}.
Furthermore they argue:
\begin{remark}\label{rem:simple_conf}
There are at most $2^{\Oh(\log^2( \makespan')+\log^2(|P|))}$ simple configurations.
\end{remark}
The better running time can be achieved by determining configurations that are equivalent to the complex configurations (via guessing and dynamic programming), guessing the support of the simple configurations, and solving the MILP with few integral variables.
The approach is a direct adaptation of the one in \cite{JKV16ICALP} for our case.
In the following, we explain the additional steps of the modified algorithm in more detail, analyse the running time  and present an outline of the complete algorithm.

We have to ensure that the makespan and the processing times are integral and that the makespan is small.
After the geometric rounding step we scale the makespan and the processing times, such that $\makespan=1/\eps^3$ and $\rmakespan=(1+\eps)/\eps^3$ holds and the processing times have the form $(1+\eps)^x\eps^2\makespan=(1+\eps)^x/\eps$.
Next we apply a second rounding step for the big processing times, setting $\breve{p}_{t,j}=\ceil{\bar{p}_{t,j}}$ for $\bar{p}_{t,j}\in\bjobsizes_t$ and denote the set of these processing times by $\breve{\bjobsizes}_t$.
Obviously we have $|\breve{\bjobsizes}_t|\leq |\bjobsizes_t| \leq  1/\eps\log(1/\eps^{2})+1$.
We denote the corresponding instance by $\breve{I}$.
Since for a schedule with makespan $\makespan$ for instance $I$ there are at most $1/\eps^2$ big jobs on any machine, we get:
\begin{lemma}\label{lem:rounded_instance_2}
If there is a schedule with makespan at most $\makespan$ for $I$, the same schedule has makespan at most $(1+2\eps)\makespan$ for instance $\breve{I}$ and any schedule for instance $\breve{I}$ can be turned into a schedule for $I$ without increase in the makespan.
\end{lemma}
We set $\rrmakespan=(1+2\eps)\makespan$ and for each machine type $t$ we consider the set of configurations $\Conf_t(\floor{\rrmakespan})$ for $\breve{\bjobsizes}_t$ with size at most $\floor{\rrmakespan}$.
Rounding down $\rrmakespan$ ensures integrality and causes no problems, because all big processing times are integral.
Furthermore let $\Conf^c_t(\floor{\rrmakespan})$ and $\Conf^s_t(\floor{\rrmakespan})$ be the subsets of complex and simple configurations.
Due to Remark \ref{rem:simple_conf} we have:
\begin{equation}\label{eq:simple_conf_bounded}
|\Conf^s_t(\floor{\rrmakespan})| \in 2^{\Oh(\log^2 \floor{\rrmakespan}+\log^2|\breve{\bjobsizes}_t|)} = 2^{\Oh(\log^2 \nicefrac{1}{\eps}))}\end{equation}
Due to Theorem \ref{thm:thin_solutions} (using the same considerations concerning configuration ILPs like in the last paragraph), we get that there is a solution $(\breve{z}_C,\breve{x}_{j,t})$ for $\MILP(\rrmakespan)$ (adjusted to this case) that uses for each machine type $t$ at most $4(|\breve{\bjobsizes}_t|+1)\log(4(|\breve{\bjobsizes}_t|+1)\floor{\rrmakespan})\in\Oh(1/\eps\log^2 1/\eps)$ many configurations from $\Conf_t(\floor{\rrmakespan})$.
Moreover, at most $2(|\breve{\bjobsizes}_t|+1)\log(4(|\breve{\bjobsizes}_t|+1)\floor{\rrmakespan})\in\Oh(1/\eps\log^2 1/\eps)$ complex configurations are used and each of them is used only once.
Since each configuration corresponds to at most $1/\eps^2$ jobs, there are at most $\Oh(1/\eps^3\log^2 1/\eps)$ many jobs for each type corresponding to complex configurations.
Hence, we can determine the number of complex configurations $m_t^c$ for machine type $t$ along with the number of jobs $k^c_{t,p}$ with processing time $p\in\breve{\bjobsizes}_t$ that are covered by a complex configuration  in $(1/\eps^3\log^2 1/\eps)^{\Oh(\nicefrac{K}{\eps}\log \nicefrac{1}{\eps})}=2^{\Oh(\nicefrac{K}{\eps}\log^2 \nicefrac{1}{\eps})}$ many steps via guessing.
Now we can use a dynamic program to determine configurations (with multiplicities) that are equivalent to the complex configurations in the sense that their size is bounded by $\floor{\rrmakespan}$, their summed up number is $m_t^c$ and they cover exactly $k^c_{t,p}$ jobs with processing time $p$.
The dynamic program iterates through $[m_t^c]$ determining ${\breve{\bjobsizes}_t}$-indexed vectors $y$ of non-negative integers with $y_p\leq k^c_{t,p}$.
A vector $y$ computed at step $i$ encodes that $y_p$ jobs of size $p$ can be covered by $i$ configurations from $\Conf_t(\floor{\rrmakespan})$.
We denote the set of configurations the program computes with $\tilde{\Conf}_t$ and the multiplicities with $\tilde{z}_C$ for $C\in\tilde{\Conf}_t$.
It is easy to see that the running time of such a program can be bounded by $\Oh(m_t^c(\prod_{p\in\breve{\bjobsizes}_t}(k^c_{t,p}+1))^2)$.
Using $m_t^c\in\Oh(1/\eps\log^2 1/\eps)$ and $ k^c_{t,p}\in\Oh(1/\eps^3\log^2 1/\eps)$ this yields a running time of $K2^{\Oh(\nicefrac{1}{\eps}\log^2 \nicefrac{1}{\eps})}$, when considering all the machine types.

Having determined configurations that are equivalent to the complex configurations, we may just guess the simple configurations.
For each machine type, there are at most $2^{\Oh(\log^2 1/\eps)}$ simple configurations and the number of configurations we need is bounded by $\Oh(1/\eps\log^2 1/\eps)$.
Therefore, the number of needed guesses is bounded by $2^{\Oh(\nicefrac{K}{\eps}\log^4 \nicefrac{1}{\eps})}$.
Now we can solve a modified version of $\MILP(\rrmakespan)$ in which $z_C$ is fixed to $\tilde{z}_C$ for $C\in\tilde{\Conf}_t$ and only variables $z_{C'}$ corresponding to the guessed simple configurations are used.
The running time for the algorithm by Lenstra and Kannan can again be bounded by $2^{\Oh(K\log K \nicefrac{1}{\eps}\log^3 \nicefrac{1}{\eps})}\poly(|I|)$.
Thus we get an overall running time of $2^{\Oh(K\log K \nicefrac{1}{\eps}\log^4 \nicefrac{1}{\eps})}\poly(|I|)$. 
Considering the two cases $2^{\Oh(K\log K \nicefrac{1}{\eps}\log^4 \nicefrac{1}{\eps})}<\poly(|I|)$ and $2^{\Oh(K\log K \nicefrac{1}{\eps}\log^4 \nicefrac{1}{\eps})}\geq\poly(|I|)$ yields the claimed running time of:
\[2^{\Oh(K\log(K) \nicefrac{1}{\eps}\log^4 \nicefrac{1}{\eps})}+\poly(|I|)\] 
Hence, the proof of the part of Theorem \ref{thm:main_result} concerning unrelated scheduling is complete.
We conclude this section with a summary of the complete algorithm.
\begin{algorithm}
\
\begin{enumerate}
\item Simplify the input via scaling, geometric rounding and a second rounding step for the big jobs with an error of $2\eps\makespan$. We now have $\makespan=1/\eps^3$.
\item Guess the number of machines $m_t^c$ with a complex configuration for each machine type $t$ along with  the number $k^c_{t,p}$ of jobs with processing time $p$ covered by complex configurations for each big processing time $p\in\breve{\bjobsizes}_t$.
\item \label{algo:step_complex}For each machine type $t$ determine via dynamic programming configurations that are equivalent to the complex configurations.
\item \label{algo:guess_simple}Guess the simple configurations used in a thin solution.
\item Build the simplified mixed integer linear program $\MILP(\rrmakespan)$ in which the variables for configurations from step \ref{algo:step_complex} are fixed and only integral variables for configurations guessed in step \ref{algo:guess_simple} are used. Solve it with the algorithm by Lenstra and Kannan.
\item If there is no solution for each of the guesses, report that there is no solution with makespan $T$.
\item Generate an integral solution for $\MILP(\rrmakespan+\eps\makespan+\eps^2\makespan)$ via a flow network utilizing flow integrality.
\item With an additional error of $\eps^2\makespan$ due to the small jobs the integral solution is turned into a schedule. 
\end{enumerate}
\end{algorithm}

\section{The Santa Claus Problem}\label{sec:santa}

Adapting the result for unrelated scheduling we achieve an EPTAS for the Santa Claus problem.
It is based on the basic EPTAS together with the second running time improvement.
In the following we show the needed adjustments.

\paragraph{Preliminaries.}

Wlog. we present a $(1-\eps)^{-1}$-approximation instead of a $(1+\eps)$-approximation.
Moreover, we assume $\eps<1$ and that $m\leq n$, because otherwise the problem is trivial.

The dual approximation method can be applied in this case as well.
However, since we have no approximation algorithm with a constant rate, the binary search is slightly more expensive.
Still we can use for example the algorithm by Bezáková and Dani \cite{BD05} to find a bound $B$ for the optimal makespan with $B\leq\Opt\leq (n-m+1)B$.
In $\Oh(\log((n-m)/\eps))$ many steps we can find a guess for the optimal minimum machine load $\makespan^*$ such that $\makespan^*\leq\Opt<\makespan^*+\eps B$ and therefore $\makespan^*>(1-\eps)\Opt$.
It suffices to find a procedure that given an instance and a guess $\makespan$ outputs a solution with objective value at least $(1-\alpha\eps)\makespan$ for some constant $\alpha$.

Concerning the simplification of the input, we first scale the makespan and the running times such that $\makespan=1/\eps^3$.
Then we set the processing times that are bigger than $\makespan$ equal to $\makespan$.
Next we round the processing times down via geometric rounding:
We set $\bar{p}_{t,j}=(1-\eps)^x\eps^2\makespan$ with $x=\ceil{\log_{1-\eps}p_{tj}/(\eps^{2} \makespan)}$.
The number of big jobs for any machine type is again bounded by $1/\eps\log(1/\eps^2)\in\Oh(1/\eps\log 1/\eps)$.
For the big jobs we apply the second rounding step setting $\breve{p}_{t,j}=\floor{\bar{p}_{t,j}}$ and denote the resulting big processing times with $\breve{\bjobsizes}_t$, the corresponding instance by $\breve{I}$ and the occurring small processing times by $\sjobsizes_t$.
The analogue of Lemma \ref{lem:rounded_instance_2} holds, i.e. at the cost of $2\eps\makespan$ we may search for a solution for the rounded instance $\breve{I}$.
We set $\rrmakespan=(1-2\eps)\makespan$.

\paragraph{MILP.}

In the Santa Claus problem it makes sense to use configurations of size bigger than $\rrmakespan$.
Let $P=\floor{\rrmakespan}+\max\sett{\breve{p}_{t,j}}{t\in[K],j\in\breve{\bjobsizes}_t}$. 
It suffices to consider configurations with size at most $P$ and for each machine type $t$ we denote the corresponding set of configurations by $\Conf_t(P)$.
Again we can bound $\Conf_t(P)$ by $2^{\Oh(\nicefrac{1}{\eps}\log^2 \nicefrac{1}{\eps})}$.
The MILP has integral variables $z_{C,t}$ for each such configuration and fractional ones like before.
The constraints (\ref{eq:MILP_Conf}) and (\ref{eq:MILP_job_assignment}) are adapted changing only the set of configurations and for constraint (\ref{eq:MILP_big_jobs_bounded_by_conf}) additionally in this case the left-hand side has to be at least as big as the right hand side.
The last constraint (\ref{eq:MILP_jobs_fit_into_space}) has to be changed more.
For this we partition $\Conf_t(P)$ into the set $\hat{\Conf}_t(P)$ of big configurations with size bigger than $\floor{\rrmakespan}$ and the set $\check{\Conf}_t(P)$ of small configurations with size at most $\floor{\rrmakespan}$.
The changed constraint has the following form:
\begin{align}
\sum_{C\in\check{\Conf}_t(P)} \size(C) z_{C,t} + \sum_{p\in\sjobsizes_t}p\sum_{j\in\rjobs_{t}(p)} x_{j,t}   &\geq (m_t-\sum_{C\in\hat{\Conf}_t(P)}z_{C,t})\rrmakespan                             & \forall t\in [K] \label{eq:MILP_santa}
\end{align}
We denote the resulting MILP by $\MILP(\rrmakespan,P)$ and get the analogue of Lemma \ref{lem:MILP_and_schedule}:
\begin{lemma}\label{lem:MILP_and_schedule_Santa}
If there is schedule with minimum machine load $\rrmakespan$, there is a feasible (integral) solution of $\MILP(\rrmakespan,P)$; and if there is a feasible integral solution for $\MILP(\rrmakespan,P)$, there is a schedule with minimum machine load at least $\rrmakespan-\eps^2\makespan$.
\end{lemma}
\begin{proof}
Let $\sigma$ be a schedule with minimum machine load $\rrmakespan$.
We first consider only the machines for which the received load due to big jobs is at most $P$.
These machines obey exactly one configuration from $ \Conf_t(P)$ and we set the corresponding integral variables like before.
The rest of the integral variables we initially set to $0$.
Now consider a machine of type $t$ that receives more than $P$ load due to big jobs.
We can successively remove a biggest job from the set of big jobs assigned to the machine until we reach a subset with summed up processing time at most $P$ and bigger than $\floor{\rrmakespan}$.
This set corresponds to a big configuration $C'$ and we increment the variable $z_{C',t}$.
The fractional variables are set like in the unrelated scheduling case and it is easy to verify that all constraints are satisfied.

Now let $(z_{C,t},x_{j,t})$ be an integral solution of MILP($\rrmakespan$).
Again we can assign the jobs to distinct machine types based on the $x_{j,t}$ variables and the configurations to machines based on the $z_{C,t}$ variables such that each machine receives at most one configuration.
Based on these configurations we can create slots for the big jobs and for each type $t$ we can successively assign big jobs until all slots are filled.
Now we can, for each type, iterate through the machines that received small configurations and greedily assign small jobs.
When the makespan $\rmakespan$ would be exceeded due to some job, we stop assigning to the current machine (not adding the current job) and continue with the next machine.
Because of (\ref{eq:MILP_santa}) we can cover all of the machines by this.
Since the small jobs have size at most $\eps^{2}\makespan$ we get a schedule with makespan at least $\rmakespan-\eps^2\makespan$.
There may be some remaining jobs that can be assigned arbitrarily.
\end{proof}

To solve the MILP we adapt the techniques by Jansen et al. \cite{JKV16ICALP}, which is slightly more complicated for the modified MILP.
Unlike in the previous section in order to get a thin solution that still fulfils (\ref{eq:MILP_santa}), we have to consider big and small configurations separately for each machine type.
Note that for a changed solution of the MILP (\ref{eq:MILP_santa}) is fulfilled, if the summed-up size of the small and the summed up number of the big configurations is not changed. 
Given a solution $(\tilde{z}_{C,t},\tilde{x}_{j,t})$ for the MILP and a machine type $t$, we set $\check{m}_t=\sum_{C\in\check{\Conf}_t(P)}\tilde{z}_{C,t}$ and $\hat{m}_t=\sum_{C\in\hat{\Conf}_t(P)}\tilde{z}_{C,t}$, and furthermore  $\check{k}_{t,p}=\sum_{C\in\check{\Conf}_t(P)}C_p \tilde{z}_{C,t}$ and $\hat{k}_{t,p}=\sum_{C\in\hat{\Conf}_t(P)}C_p \tilde{z}_{C,t}$ for $p\in\breve{\bjobsizes}_t$.
We get two configuration ILPs:
The first is given by $\check{m}_t$, $\breve{\bjobsizes}_t$, $\check{k}_{t}$ and $\check{\Conf}_t(P)$ and we call it the \emph{small} ILP.
The second is given by $\hat{m}_t$, $\breve{\bjobsizes}_t$, $\hat{k}_{t}$ and $\hat{\Conf}_t(P)$ and we call it the \emph{big} ILP.
For the small ILP the set of configurations is given by the upper bound $\floor{\rrmakespan}$ on the configuration size and we define the simple and complex configurations accordingly denoting them by $\check{\Conf}^s(P)$ and $\check{\Conf}^c(P)$ respectively.
We can directly apply Theorem \ref{thm:thin_solutions} to the small ILP like before without changing the summed-up size of the small configurations.
This is not the case for the big ILP because in this case the set of configurations is defined by an upper and lower bound for the configuration size and hence Theorem \ref{thm:thin_solutions} can not be applied directly.
Note that considering the set of configurations given just by the upper bound $P$ is not an option, since this could change the number of big configurations that are used. 
However, when looking more closely into the proof of Theorem \ref{thm:thin_solutions} given in \cite{JKV16ICALP}, it becomes apparent that the result can easily be adapted.
For this we call a configuration $C$ in this case simple if $|\supp(C)|\leq\log(P+1)$ and complex otherwise and denote the corresponding sets by $\hat{\Conf}^s(P)$ and $\hat{\Conf}^c(P)$ respectively.
Without going into details we give the outline how the proof can be adjusted to this case:

The main tools in the proof are variations of Theorem \ref{thm:Eisenbrandt_Shmonin} and the so called Sparsification Lemma.
Theorem \ref{thm:Eisenbrandt_Shmonin} actually works with any set of configurations and therefore we can restrict its use to big configuration. 
Moreover, the Sparsification Lemma is used to exchange complex configurations that are used multiple times with configurations that have a smaller support but the same size.
Therefore big configurations are exchanged only with other big configurations.
Moreover, the Sparsification Lemma still holds when considering a set of configurations with a lower and upper bound for the size.

Hence, there is a thin solution for the big ILP and obviously the summed-up number of configurations stays the same.
Summarizing we get:
\begin{corollary}
If MILP$(\breve{T})$ has a solution, there is also a solution $(z_{C,t},x_{j,t})$ such that for each machine type $t$:
\begin{enumerate}
\item $|\supp(y|_{\check{\Conf}^c_t(P)})|\leq 2(|\breve{\bjobsizes}_t|+1)\log(4(|\breve{\bjobsizes}_t|+1)\floor{\rrmakespan})$, $|\supp(y|_{\hat{\Conf}^c_t(P)})|\leq 2(|\breve{\bjobsizes}_t|+1)\log(4(|\breve{\bjobsizes}_t|+1)P)$ and $z_{C,t}\leq 1$ for $C\in\check{\Conf}^c_t(P)\cup\hat{\Conf}^c_t(P) $.
\item $|\supp(z_t)|\leq 4(|\breve{\bjobsizes}_t|+1)( \log(4(|\breve{\bjobsizes}_t|+1)\floor{\rrmakespan}) + \log(4(|\breve{\bjobsizes}_t|+1)P))$.
\end{enumerate}
\
\end{corollary}
Note that like before the terms above can be bounded by $\Oh(1/\eps\log^2 1/\eps)$.
Utilizing this corollary we can again solve the MILP rather efficiently.
For this we have to guess the numbers $\check{m}_t^c$ and $\hat{m}_t^c$ of machines that are covered by small and big complex configurations respectively.
In addition we guess like before the numbers of big jobs corresponding to the complex configurations.
With this we can determine via dynamic programming suitable configurations.
For the small configurations we can use the same dynamic program as before and for the second one we can use a similar one that guarantees that we find big configurations.
In the MILP we fix the big configurations we have determined and guess the non-zero variables corresponding to the simple configurations.
Although this procedure is a little bit more complicated than in the unrelated machine case, the bound for the running time remains the same.

\paragraph{Rounding.}

To get an integral solution of the MILP we build a similar flow network.
However in this case  $\eta_{t,p}=\floor{\sum_{j\in\rjobs_{t}(p)}x_{j,t}}$ is set to be the rounded \emph{down} (fractional) number of jobs with processing time $p$ that are assigned to machine type $t$.
We get $\eta_{t,p}\geq\sum_{C\in\Conf(T)}C_{\ell}z_{C,t}$ for big processing times $p$.
The flow network looks basically the same, with one important difference: 
The $(u_{t,p},\sink)$ have a \emph{demand} of $ \eta_{t,p}$ and an capacity of $\infty$.
We may introduce demands of $0$ for all the other edges.
The analogue of Lemma \ref{lem:max_flow} holds, that is, the flow network has a (feasible) maximum flow with value $n$.
Given such a flow we can build a new solution for the MILP changing the $x_{j,t}$ variables based on the flow decreasing the load due to small jobs by at most $\eps T + \eps^2 T$.

Flow networks with demands can be solved with a two-phase approach that first finds a feasible flow and than augments the flow until a max flow is reached.
The first problem can be reduced to a max flow problem without demands in a flow network that is rather similar to the original one with at most two additional nodes and $\Oh(|V|)$ additional edges.
Flow integrality still can be used.
For details we refer to \cite{AMO93network}.
The running time again can be bounded by $\Oh(Kn^2)$. 
Hence the overall running time of the algorithm is $2^{\Oh(K\log(K) \nicefrac{1}{\eps}\log^4 \nicefrac{1}{\eps})}+\poly(|I|)$, which concludes the proof of Theorem \ref{thm:main_result}.

\section{Uniform Machinetypes}\label{sec:uniform}

We consider the problem of unrelated scheduling with a constant number $K$ of uniform machine types.
In this version of the problem the input is as follows:
Each job has a size $p_{tj}$ for each uniform machine type $t$ and each machine $i$ has a speed value $s_i$ and a type $t_i$. 
The processing time of job $j$ on machine $i$ is given by $p_{ij} = p_{t_ij}/s_i$.

We present an EPTAS and it has the same basic structure as the ones presented so far.
However, both the MILP and its rounding are considerably more complicated and can be seen as a combination of the techniques from Section \ref{sec:basic_eptas} with ideas from~\cite{Jan10QCmaxEPTAS}.
Note that in this section we have taken less effort to get a small running time in order to keep the presentation of the result less technical.

We set $\machs_t = \sett{i\in \machs}{t_i=t}$ for each $t\in [K]$ and $ s^{(t)}_{\max} = \max\sett{s_i}{i\in\machs_t}$. 
In the following, we refer to uniform machine types as machine types or just types.

\paragraph{Preliminaries.}

Again, we may assume that a target makespan $T$ for instance $I$ is given and we employ geometric rounding to both the job sizes and machine speeds. 
More precisely, if a job $j$ has a size bigger than $T s^{(t)}_{\max}$ for a machine type $t\in[K]$, we set $\bar{p}_{tj}=\infty$.
For each job $j$ with $p_{tj}<\infty$, we set $\bar{p}_{tj}=(1+\eps)^x\eps^2 T s^{(t)}_{\max}$ with $x=\ceil{\log_{1+\eps}(p_{tj}/(\eps^2 T s^{(t)}_{\max}))}$.
Moreover, we set $\bar{s}_{i} = s^{(t)}_{\max}/(1+\eps)^y$ with $y =\ceil{\log_{1+\eps}(s^{(t)}_{\max}/s_i)}$ and call the rounded instance $\bar{I}$.
\begin{lemma}\label{lem:uni_rounded_instance}
If there is a schedule with makespan at most $\makespan$ for $I$, the same schedule has makespan at most $(1+\eps)^2\makespan$ for instance $\bar{I}$ and any schedule for instance $\bar{I}$ can be turned into a schedule for $I$ without increase in the makespan.\qed
\end{lemma}
Therefore, it suffices to search for a schedule for instance $\bar{I}$ with makespan $\bar{\makespan} := (1+\eps)^2\makespan$.
For the sake of simplicity, we do not use the $(\bar{\,\cdot\,})$-notation in the following, i.e., we assume that the instance is already rounded and the makespan properly increased.

We fix some notation:
A job size $p$ is called \emph{huge} for a speed $s$, if $p > T s$; \emph{big}, if $p\leq T s$ and $p > \eps^2 T s$; and \emph{small} otherwise.
We will not consider assigning jobs on machines for whose speeds they are huge. 
For each machine type $t$, we denote the set of occurring speeds $\sett{s_i}{i\in\machs_t}$ by $\speeds_t$; the set of machines of type $t$ and speed $s$ by $\machs_{t,s}$; and set $m_{t,s}:=|\machs_{t,s}|$. 
For each machine type $t$ and speed $s$, let $\sjobsizes_{t,s}$ and $\bjobsizes_{t,s}$ be the sets of occurring small and big processing times.
Furthermore, let $\jobsizes_t$ be the set of all occurring job sizes for type $t$.
Like before, we have $|\bjobsizes_{t,s}|\in\Oh(1/\eps\log 1/\eps)$.
For any processing time $p$ we denote the set of jobs $j$ with $p_{tj}=p$ by $\rjobs_{t}(p)$.

\paragraph{Separation of Machines.}

We will consider configurations for each machine type $t$ and speed value $s\in\speeds_t$.
However, the number of distinct speed values could be dependent in $m$ and we can not effort to introduce integral variables in the MILP for each of them.
Instead, we will introduce integral variables only for the fastest speeds of each type and round the fractional variables.
For the rounding approach, we will need a constant number of machines that receive some load from the slow speeds, and furthermore the speeds of these machines have to be faster than the slow speeds by some constant factor.
This leads to a separation of the machines into three groups $\sgroup{t}{i}$ for $i\in[3]$ for each machine type $t$.
This is done in a way, such that for $j>i$ the machines in group $\sgroup{t}{i}$ are faster than the ones in group $\sgroup{t}{j}$.
For $ i\in[3]$ and $ \mathrm{opt} \in \set{\min,\max}$, we set $s^{(t)}_{i,\mathrm{opt}} := \mathrm{opt}\sett{s_i}{i\in\sgroup{t}{i}}$.
The partition is defined by two parameters.
The first parameter \[\paramcard := \max\sett{|\bjobsizes_{t,s}| + 1}{\forall t\in[K],s\in\speeds_t}\in \Oh(1/\eps\log 1/\eps)\] controls the number of machines in the first group and the second $\paramgap := 1/\eps^2$ the speed-gap between the first and the third group.
More precisely:
\begin{itemize}
\item $\sgroup{t}{1}$ contains the $\paramcard$ fastest machines of type $t$.
\item $\sgroup{t}{2}$ contains all machines of type $t$ that are not contained in $\sgroup{t}{1}$ and whose speed is bigger than $\paramgap s^{(t)}_{1,\min}$.
\item $\sgroup{t}{3}$ contains the rest of the machines of type $t$.
\end{itemize}
Note that $\sgroup{t}{2}$ and $\sgroup{t}{3}$ might be empty.
We denote the occurring speeds in group $\sgroup{t}{i}$ by $\speeds_{t,i}$ and call the speeds from $\speeds_{t,1}\cup\speeds_{t,2}$ \emph{fast} and the rest \emph{slow}.
With these definitions we have $(\speeds_{t,1}\cup\speeds_{t,2})\cap\speeds_{t,3}=\emptyset$ and $|\speeds_{t,1}|,|\speeds_{t,2}|\in\Oh(1/\eps\log(1/\eps))$, i.e., the fast and slow speed values are distinct and we have only a constant number of fast speed values.

\paragraph{MILP.}

For each machine type $t$ and speed $s\in\speeds_t$ we consider the set $\Conf_t(s\makespan)$ of configurations $C$ for the big processing times $\bjobsizes_{t,s}$ and with $\size(C)\leq s\makespan$.
Note that $|\Conf_t(s\makespan)|\in 2^{\Oh(\nicefrac{1}{\eps}\log^2 \nicefrac{1}{\eps})}$ and therefore $|\bigcup_{s\in\speeds_{t,1}\cup\speeds_{t,2}}\Conf_t(s\makespan)|\in 2^{\Oh(\nicefrac{1}{\eps}\log^2 \nicefrac{1}{\eps})}$.

The MILP formulation in this scenario follows the same basic ideas, but is more complicated than before.
We assign jobs fractionally to machines types.
For the fast machine speeds we chose configurations integrally and for the slow ones fractionally.
Furthermore, we fractionally assign job sizes to machine speeds for which they are small. 
Lastly, we count the number of jobs of each job size that are assigned to each machine type.
If the job size is big on some fast machine of the machine type, we require an integral number of jobs.
More precisely, we introduce the following variables:
\begin{itemize}
\item Configuration variables $z^{(t,s)}_{C}$ for each machine type $t\in[K]$, occurring speed $s\in\speeds_t$ and configuration $C\in\Conf_t(s\makespan)$.
If $s$ is fast, we require $z^{(t,s)}_{C}\in\ZZ_{\geq 0}$ and otherwise $z^{(t,s)}_{C}\geq 0$.
\item Job assignment variables $x_{j,t}\geq 0$ for each machine type $t\in[K]$ and job $j\in\jobs$.
\item Job size assignment variables $y^{(t)}_{p,s}\geq 0$ for each machine type $t\in[K]$, speed $s\in\speeds_t$ and job size $p\in\sjobsizes_{t,s}$.
\item Counting variables $u^{(t)}_{p}$ for each machine type $t\in[K]$ and job size $p\in\jobsizes_t$.
If there is a fast speed $s\in\speeds_{t,1}\cup\speeds_{t,2}$, such that $p\in\bjobsizes_{t,s}$ we require $u^{(t)}_{p}\in\ZZ_{\geq 0}$ and otherwise $u^{(t)}_{p}\geq 0$.
\end{itemize}
Now the MILP is given by the above variables and the following constraints:
\begin{align}
\sum_{C\in\Conf_t(s\makespan)}z^{(t,s)}_{C}     &= m_{t,s}      & \forall t\in [K],s\in\speeds_t \label{eq:uni_MILP_Conf} \\
\sum_{t\in[K]} x_{j,t}     &= 1     & \forall j\in\jobs \label{eq:uni_MILP_job_assignment} \\
\sum_{j\in\rjobs_{t}(p)}x_{j,t}     &\leq u^{(t)}_{p}    & \forall t\in [K],p\in\jobsizes_t \label{eq:uni_MILP_big_jobs_bounded_by_conf}\\
\sum_{s:p\in\bjobsizes_{t,s}}\sum_{C\in\Conf_t(s\makespan)} C_{p}z^{(t,s)}_{C} + \sum_{s: p\in\sjobsizes_{t,s}}y^{(t)}_{p,s}    &\geq u^{(t)}_{p}    &\forall t\in [K],p\in\jobsizes_t \label{eq:uni_MILP_cardinality_job_sizes}\\
\sum_{C\in\Conf_t(s\makespan)} \size(C) z^{(t,s)}_{C} + \sum_{p\in\sjobsizes_{t,s}} p y^{(t)}_{p,s}    &\leq m_{t,s}sT     & \forall t\in [K],s\in\speeds_t \label{eq:uni_MILP_jobs_fit_into_space}
\end{align}
The constraints (\ref{eq:uni_MILP_Conf}) and (\ref{eq:uni_MILP_job_assignment}) are very similar to constraints for the other MILPs that we consider.
For each machine type it is ensured with the constraints (\ref{eq:uni_MILP_big_jobs_bounded_by_conf}) and (\ref{eq:uni_MILP_cardinality_job_sizes}) that the summed up number of jobs of each size is covered by the the chosen configurations and the small job assignments.
Furthermore, (\ref{eq:uni_MILP_jobs_fit_into_space}) guarantees that the overall processing time of the configurations and small jobs assigned to a machine speed for each type does not exceed the available area.
\begin{lemma}\label{lem:uni_MILP_and_schedule}
If there is schedule with makespan $\makespan$, there is a feasible (integral) solution of $\MILP(\makespan)$; and if there is a feasible integral solution for $\MILP(\makespan)$, there is a schedule with makespan at most $(1+\eps^2)T$.
\end{lemma}
\begin{proof}
Given a schedule $\sigma$ with makespan $T$, each machine of type $t$ with speed $s$ obeys exactly one configuration from $\Conf_t(s\makespan)$ and we can set the variables $z_{C,t}$ accordingly.
Furthermore, we set $x_{j,t_{\sigma(j)}}=1$, $x_{j^*,t}=0$ for $t\neq t_{\sigma(j)}$, $u^{(t)}_p := |\sett{j}{t_{\sigma(j)}=t, p_{t,j}=p}|$, and $y^{(t)}_{p,s} := |\sett{j}{t_{\sigma(j)}=t,s_{\sigma(j)}=s, p_{t,j}=p}|$.
It is easy to check that all conditions are fulfilled.

Like we did in the proof of Lemma \ref{lem:MILP_and_schedule}, given an integral solution $(z,x,y,u)$ we can assign the jobs to machine types, and configurations to machines.
Moreover, based one the $y^{(t)}_{p,s}$ variables we can assign jobs of size $p$ that are assigned to type $t$ to machines of speed $s$ on which they are small.
Because of (\ref{eq:uni_MILP_big_jobs_bounded_by_conf}) and (\ref{eq:uni_MILP_cardinality_job_sizes}) this can be done such that the remaining jobs of size $p$ can be scheduled into slots provided by configurations.
At this point each unscheduled job is assigned to a type and a speed.
Utilizing (\ref{eq:uni_MILP_jobs_fit_into_space}), these jobs can be scheduled greedyly with an additive error of $\eps^2T$.
\end{proof}
\begin{lemma}
The MILP has $\Oh(K(n+m))$ many constraints, $\Oh(Knm) + m2^{\Oh(\nicefrac{1}{\eps}\log^2 \nicefrac{1}{\eps})}$ many variables and $K2^{\Oh(\nicefrac{1}{\eps}\log^2 \nicefrac{1}{\eps})}$ integral variables. It can be solved in time: \[2^{\Oh(K\log(K)\nicefrac{1}{\eps^3} \log^5\nicefrac{1}{\eps})}\poly(|I|)\]
\end{lemma}
\begin{proof}
The bounds for the number of constraints and variables are easy to verify using the above considerations as well as $|\speeds_t|\leq m $ and $|\jobsizes_t|\leq n$.
The running time can be achieved with the first approach presented in Section \ref{sec:better_running_time}:
Using Theorem \ref{thm:Eisenbrandt_Shmonin}, we can argue that $\Oh(1/\eps\log^2(1/\eps))$ many integral variables for each machine type $t$ and speed $s$ suffice. 
Therefore the number of needed guesses is $2^{\Oh(\nicefrac{K}{\eps^3}\log^5 \nicefrac{1}{\eps})}$. 
Running the algorithm of Lenstra and Kannan with $ \Oh(K/\eps^2\log^3(1/\eps)) $ many integral variables takes $2^{\Oh(K\log(K)\nicefrac{1}{\eps^2} \log^4\nicefrac{1}{\eps})}\poly(|I|)$ time.
Together we get the stated running time.
\end{proof}

\paragraph{Rounding.}

We present rounding approaches for all fractional variables and start with the configuration variables.

\subparagraph{Configuration Variables.}

We fix a type $t$ and slow speed $s\in\speeds_t$ and set $k_p := \sum_{C\in\Conf_t(s\makespan)} C_{p}z^{(t,s)}_{C}$ for each $p\in\bjobsizes_{t,s}$.
We have:
\begin{align}
\sum_{C\in\Conf_t(s\makespan)} z^{(t,s)}_{C} &= m_{t,s} &\\
\sum_{C\in\Conf_t(s\makespan)} C_{p}z^{(t,s)}_{C} &= k_p & \forall p\in\bjobsizes_{t,s}
\end{align}
It is easy to check that, if we replace the solution of this LP with any other solution and change the MILP solution accordingly, the resulting MILP solution will still be feasible. 
We transform the solution into a basic feasible solution.
This can be done in time polynomial in $1/\eps$ and $|I|$.
The LP has $|\bjobsizes_{t,s}| + 1$ many constraints and therefore the solution has at most $|\bjobsizes_{t,s}| + 1$ many variables greater than $0$.
Now the idea, is to round down the fractional values and to assign the respective job sizes that lost covering by the configurations to the fastest group $\sgroup{t}{1}$.
More precisely, we chose some injective mapping $\xi$ between the configurations $C$ with fractional variables $z^{(t,s)}_{C}$ and the machines from $\sgroup{t}{1}$.
This can be done, due to the choice of the parameter $\paramcard$ that regulates the number of machines in $\sgroup{t}{1}$.
Now, we round down $z^{(t,s)}_{C}$ to the next integral value and increase $y^{(t)}_{p,s_{\xi(C)}}$ by $(z^{(t,s)}_{C}-\floor{z^{(t,s)}_{C}})C_p \leq C_p$ for each $p\in\bjobsizes_{t,s}$.
We perform these steps for all types $t$ and slow speeds $s\in\speeds_t$.
Note that any particular variable $y^{(t)}_{p,s}$ might be increased several times for each speed value for which $p$ is big.
Let $\tilde{y}^{(t)}_{p,s}$ denote the resulting increased $y^{(t)}_{p,s}$ variables and $\bar{z}^{(t,s)}_{C}$ the resulting configuration variables.
For $(\bar{z},x,\tilde{y},u)$ the constraints (\ref{eq:uni_MILP_Conf})-(\ref{eq:uni_MILP_big_jobs_bounded_by_conf}) obviously still hold and it is easy to see that this is also the case for (\ref{eq:uni_MILP_cardinality_job_sizes}), while (\ref{eq:uni_MILP_jobs_fit_into_space}) might be violated for speeds associated with the fastest machine groups.
We show that a modified version of (\ref{eq:uni_MILP_jobs_fit_into_space}) still holds.

Consider a machine $i\in\sgroup{t}{1}$. 
For each slow speed value $s\in\speeds_{t,3}$ there may be one configuration $C$ that is mapped to $i$.
The summed up job sizes that are reassigned to $s_i$ because of this are bounded by $sT$.
Summing up over all speed values $s\in\speeds_{t,3}$ and utilizing the convergence of the geometric series, the rounding of the speed values, and the fact that $ s^{(t)}_{3,\max} \leq \gamma s^{(t)}_{1,\max} = 1/\eps^2s^{(t)}_{1,\max}$, we get:
\begin{align*}
\sum_{s\in\speeds_{t,3}}sT &= Ts^{(t)}_{3,\max} \sum_{s\in\speeds_{t,3}} \frac{s}{s^{(t)}_{3,\max}}< Ts^{(t)}_{3,\max} \sum^{\infty}_{i=0} \frac{1}{(1+\eps)^i}\\
 & \leq Ts^{(t)}_{1,\min}(\eps+\eps^2) \leq Ts_i(\eps+\eps^2)
\end{align*}
Hence, (\ref{eq:uni_MILP_jobs_fit_into_space}) holds if we increase the makespan on the right hand side by $(\eps + \eps^2)T$.

\subparagraph{Counting Variables.}

The rounding step for the counting variables $u$ is the easiest: 
We round them up and assign the extra job sizes to the fastest machine speed in the group, that is, for each $t\in[K]$ and $p\in \jobsizes_t$ we set $\bar{u}_p^{(t)} = \ceil{u_p^{(t)}}$ and increase $\tilde{y}^{(t)}_{p,s^*}$ by $\bar{u}_p^{(t)} - u_p^{(t)} \leq 1$, where $s^* = s^{(t)}_{\max}$. 
We denote the changed $\tilde{y}^{(t)}_{p,s^*}$ by $\breve{y}^{(t)}_{p,s^*}$
Again, it is easy to see that for $(\bar{z},x,\breve{y},\bar{u})$ the constraints (\ref{eq:uni_MILP_Conf})-(\ref{eq:uni_MILP_cardinality_job_sizes}) still hold, while (\ref{eq:uni_MILP_jobs_fit_into_space}) is violated.
However, we can bound the increase the fastest speed receives, again utilizing the geometric series:
\[ \sum_{p\in\sjobsizes_{t,s^*}} p(\breve{y}^{(t)}_{p,s^*}-\tilde{y}^{(t)}_{p,s^*}) \leq \sum_{p\in\sjobsizes_{t,s^*}} p \leq \eps^2 s^* T \frac{1+\eps}{\eps} = (\eps + \eps^2)s^*T\]
Hence, (\ref{eq:uni_MILP_jobs_fit_into_space}) holds if we further increase the makespan by $(\eps + \eps^2)T$.

\subparagraph{Job Size Assignment Variables.}

Consider the constraint (\ref{eq:uni_MILP_cardinality_job_sizes}) for the solution $(\bar{z},x,\breve{y},\bar{u})$.
Note that the right hand side and the first sum on the left hand side are both integral.
Therefore, we can scale the $\breve{y}^{(t)}_{p,s}$ variables down, such that $\sum_{s: p\in\sjobsizes_{t,s}}\breve{y}^{(t)}_{p,s}$ is integral for each machine type $t$ and job size $p\in\jobsizes_t$ and (\ref{eq:uni_MILP_cardinality_job_sizes}) is still fulfilled.
We fix a machine type $t$, set $k_p := \sum_{s: p\in\sjobsizes_{t,s}}\breve{y}^{(t)}_{p,s}$ for each $p\in\jobsizes_t$ and assume $k_p\in\ZZ$, because of the argument above.
Furthermore, we set $L_{t,s} := \sum_{p\in\sjobsizes_{t,s}} p \breve{y}^{(t)}_{p,s}$ for each $s\in\speeds_t$.
With these definitions, we have:
\begin{align}
\sum_{s: p\in\sjobsizes_{t,s}}\breve{y}^{(t)}_{p,s}   &=   k_p   & \forall p\in\jobsizes_t \label{eq:rounding_jobs}\\
\sum_{p\in\sjobsizes_{t,s}} p \breve{y}^{(t)}_{p,s}   &=   L_{t,s}   & \forall s\in\speeds_t
\end{align}
If we replace the values $\breve{y}^{(t)}_{p,s}$ with any other solution for the above LP, we get an equivalent MILP solution. 
We can use a variation of the classical rounding approach by Lenstra, Shmoys and Tardos \cite{LST90} to transform the solution $\breve{y}^{(t)}_{p,s}$.

For the sake of completeness, we summarize the main ideas of the rounding.
The solution is transformed into a basic feasible one and the following bipartite graph is considered.
There are two types of nodes, some associated with sizes $p$ and some with speeds $s$.
For any $p$ or $s$, there can be at most one node; there are such nodes if and only if there are fractional variable $\breve{y}^{(t)}_{p,s'}$ or $\breve{y}^{(t)}_{p',s}$ left; and they are connected with an edge, if there is a fractional variable $\breve{y}^{(t)}_{p,s}$.
Using a counting argument and some further considerations, it can be shown that this graph is a pseudoforest, i.e., all connected components are either trees or trees with one extra edge.
Furthermore, because $k_p$ is integral, the definition of the graph, together with the constraint (\ref{eq:rounding_jobs}), yield that all the leafs are associated to speeds.
Using this structure, we can define an injective mapping $\xi$ from the job sizes for which there is a fractional variable $\breve{y}^{(t)}_{p,s}$ to the speeds such that $\breve{y}^{(t)}_{p,\xi(s)}$ is one of the fractional variables.
This can be done as follows:
For each connected component there may be at most one cycle in the graph with alternating size and speed nodes and a suitable injective mapping for the corresponding sizes and speeds can easily be found, by going around the cycle and appropriately mapping consecutive nodes.
After removing the corresponding nodes and edges, only trees remain in the graph.
For each tree we can chose an arbitrary leaf.
The leaf corresponds to a speed and its neighbor to a size and we can map the size to the speed and remove both corresponding nodes from the graph.
Iterating this yields the mapping $\xi$.
All the above steps can be performed in polynomial time in $1/\eps$ and $|I|$.

We use the mapping $\xi$ to round the variables $\breve{y}^{(t)}_{p,s}$ and because $\xi$ is injective we can guarantee that each speed receives at most one extra small job.
More precisely, for each $s\in\speeds_t$ we set $\bar{y}^{(t)}_{p,s} = \ceil{\breve{y}^{(t)}_{p,s}}$, if $\xi(p)=s$ and  $\bar{y}^{(t)}_{p,s} = \floor{\breve{y}^{(t)}_{p,s}}$ otherwise. 
For the solution $(\bar{z},x,\bar{y},\bar{u})$ the constraints (\ref{eq:uni_MILP_Conf})-(\ref{eq:uni_MILP_cardinality_job_sizes}) still hold, while for (\ref{eq:uni_MILP_jobs_fit_into_space}) the makespan has to be increased further by $\eps^2T$.

\subparagraph{Job Assignment Variables.}

The rounding of the job assignment variables is the same as in the regular machine types case. 
The only difference is that we can set $\eta_{t,p} = \bar{u}^{(t)}_p$ in this case.
Since all the values $\bar{u}^{(t)}_p$ are integral in this case there is no rounding error in this step.
Let $\bar{x}_{j,t}$ be the rounded version of $x_{j,t}$.

Summarizing the rounding steps, for the solution $(\bar{z},\bar{x},\bar{y},\bar{u})$ the constraints (\ref{eq:uni_MILP_Conf})-(\ref{eq:uni_MILP_cardinality_job_sizes}) hold together with:
\begin{align}
\sum_{C\in\Conf_t(s\makespan)} \size(C) \bar{z}^{(t,s)}_{C} + \sum_{p\in\sjobsizes_{t,s}} p \bar{y}^{(t)}_{p,s}    &\leq m_{t,s}s(1 + 2\eps + 3\eps^2)T     & \forall t\in [K],s\in\speeds_t
\end{align}

\paragraph{Analysis.}

Summarizing the above steps, we can construct a schedule with makespan at most $(1+\eps)^2(1 + 2\eps + 4\eps^2)T\leq(1 + 27\eps)T$ (assuming a schedule with makespan $T$ exists), by building and solving the MILP, then rounding it and lastly transforming it into a schedule like in the proof of Lemma \ref{lem:uni_MILP_and_schedule}.
Solving the MILP is again the most expensive step and with a simple case analysis we get a running time of:
\[2^{\Oh(K\log(K)\nicefrac{1}{\eps^3} \log^5\nicefrac{1}{\eps})} + \poly(|I|)\]

\section{Vector Scheduling}\label{sec:vector}

We present an EPTAS for $D$-dimensional unrelated vector scheduling, where both the dimension $D$ and the number $K$ of machine types are constant.
In this problem variant for each job $j$ a $D$-dimensional processing time vector $p_{tj}=(p^{(1)}_{tj},\dots, p^{(D)}_{tj})$ is given and the makespan is defined as the maximum load any machine receives in any dimension, i.e., $\Cmax(\sigma)  =\max_{i\in\machs}\|\sum_{j\in\sigma^{-1}(i)} p_{ij}\|_{\infty}$.
We define $\jobsizes = \sett{p^{(d)}_{tj}}{d\in[D],t\in[K],j\in\jobs}$ and $\vjobsizes = \sett{p_{tj}}{t\in[K],j\in\jobs}$.

The EPTAS is a direct adaptation of the one for the one dimensional case.
In the following we briefly describe the needed extra steps and modification.
Note, that we consider this result to be a proof of concept and took little effort to optimize the running time.

\paragraph{Preliminaries.}

We again use the dual approximation approach to get a guess $T$ of the makespan.
As an upper bound for this we can use the schedule that we get by assigning each job $j$ to a machine $i$ where $\sum_{d=1}^{D} p^{(d)}_{ij}$ is minimal.
It is easy to see that this approach yields a $Dm$-approximation and we can use this result for the dual approximation like described in Section \ref{sec:santa}.

First, we perform rounding steps similar to those for the other results.
For each $p_{tj}\in\vjobsizes$ with $p^{(d)}_{tj} > T$ in at least one dimension $d$ we set $\bar{p}_{tj} = (\infty,\dots,\infty)$ and for all other processing time vectors $p_{tj}$ we apply geometric rounding.
Let $\paramthresh = (\eps^2/D)^D$ be some threshold parameter.
We set $\bar{p}^{(d)}_{tj} =  (1+\eps)^x\paramthresh T$ with $x = \ceil{\log_{1+\eps} (p^{(d)}_{tj}/(\paramthresh T))}$ yielding a rounded vector $\bar{p}_{tj}$ and a corresponding rounded instance $\bar{I}$.

For a given processing time vector the numbers that can occur in the different dimensions may still differ strongly.
This complicates the problem, but we can reduce the extra complexity to some degree via a second rounding step:
For each $\bar{p}_{tj}$ we set $\tilde{p}^{(d)}_{tj} = \max\set{\bar{p}^{(d)}_{tj}, \|\bar{p}_{tj}\|_{\infty}\eps/D}$ yielding a rounded vector $\tilde{p}_{tj}$ and a corresponding rounded instance $\tilde{I}$.
Similar rounding steps were used by Chekuri and Khanna \cite{chekuri2004multidimensional} and Bonifaci and Wiese \cite{BW12}.
\begin{lemma}
If there is a schedule with makespan at most $\makespan$ for $I$, the same schedule has makespan at most $(1+\eps)^2\makespan$ for instance $\tilde{I}$ and any schedule for instance $\tilde{I}$ can be turned into a schedule for $I$ without increase in the makespan.
\end{lemma}
\begin{proof}
Consider a schedule $\sigma$ with makespan $T$ for $I$.
The first rounding step may increase the makespan by a factor of $(1+\eps)$.
We fix a machine $i$, and a dimension $d$ and bound the increase in load on machine $i$ in dimension $d$ for instance $\tilde{I}$. 
Let $j$ be a job with $\sigma(j)=i$.
If $\bar{p}^{(d)}_{ij} = \|\bar{p}_{ij} \|_{\infty}$, then job $j$ causes no extra load on $i$ in dimension $d$ and if $\bar{p}^{(d')}_{ij} = \|\bar{p}_{ij} \|_{\infty}$ for some dimension $d'\neq d$, there might be an increase of at most $\|\bar{p}_{ij} \|_{\infty}\eps/D$.
In fact, the summed up load machine $i$ receives in dimension $d'$ might increase the load in $d$ by an $\eps/D$-factor in this fashion.
Because the load in dimension $d'$ is bounded by $(1+\eps)T$ and there are $D-1$ dimensions $d'\neq d$, the overall load increase in dimension $d$ on $i$ can be up to $(D-1)(1+\eps)T\eps/D\leq \eps (1+\eps)T$.
\end{proof}
Hence, we may search for a schedule for instance $\tilde{I}$ with makespan $\tilde{\makespan} := (1+\eps)^2\makespan$.
For the sake of simplicity, we do not use the $(\tilde{\,\cdot\,})$-notation in the following, i.e., we assume that the instance is already rounded and the makespan properly increased.

In this context we call a size $q\in\jobsizes$ \emph{big}, if $q > \paramthresh T$ and \emph{small} otherwise.
Furthermore, we call a processing time vector $p\in\vjobsizes$ \emph{big}, if there is a dimension $d\in[D]$, such that $p^{(d)}$ is big, and \emph{small} otherwise.
Because of the second rounding step, we have $p^{(d)} > \paramthresh T\eps/D$ for each big vector $p$ and dimension $d$. 
Let $\bjobsizes_t$ and $\sjobsizes_t$ be the sets of big and small processing time vectors occurring on machine type $t$.
Note that $|\bjobsizes_t| \leq (\ceil{\log_{1+\eps}(D/(\paramthresh\eps))})^D \leq (3D/\eps \log(D/\eps) )^D$.
Using these definition, the bound on the number of big jobs is much bigger than in the other cases.
We chose this definition, because in the rounding of the MILP solution, each machine may receive a big (but constant) number of jobs for each small job size and to bound the overall load the small jobs have to be appropriately small. 

For each processing time vector $p\in\vjobsizes$ we denote the set of jobs $j$ with $p_{tj} = p$ with $\rjobs_t(p)$.

\paragraph{MILP.}

Similar to the one dimensional case, for any set $V$ of processing time vectors we call the  $V$-indexed vectors of non-negative integers $\ZZ_{\geq 0}^V$ configurations (for $V$), set the size $\size(C)$ of a configuration $C$ to be the corresponding vector of sizes, i.e., $\size(C) = \sum_{p\in\vjobsizes}C_pp$, and set $\Conf_t(T)$ to be the set of configurations $C$ for $\bjobsizes_t$ with $\size(C) \leq T^D$.
Note that: 
\[|\Conf_t(T)| \leq (\frac{D}{\theta\eps}+1)^{(3D/\eps \log(D/\eps) )^D} \leq (\frac{D}{\eps})^{3D(\nicefrac{3D}{\eps} \log\nicefrac{D}{\eps} )^D}\leq 2^{(\nicefrac{3D}{\eps}\log \nicefrac{D}{\eps})^{D+1}} \]

The MILP is a straight-forward adaptation of the one for the one-dimensional case with one important difference:
The jobs are fractionally assigned to configurations belonging to a type, instead of just being assigned to machine types.
More precisely, we introduce integral variables $z_{C,t}\in\ZZ_{\geq 0}$ for each machine type $t\in[K]$ and configuration $C\in\Conf_t(T)$, and fractional variables $x_{j,t,C}\geq 0$ for each job $j\in\jobs$, machine type $t\in[K]$ and configuration $C\in\Conf_t(T)$.
For $p_{tj}=\infty^D$ we set $x_{j,t,C}=0$. 
$\MILP(T)$ is given by:
\begin{align}
\sum_{C\in\Conf_t(T)}z_{C,t} &= m_t & \forall t\in [K] \label{eq:vMILP_Conf} \\
\sum_{t\in[K]} \sum_{C\in\Conf_t(T)} x_{j,t,C} &= 1 & \forall j\in\jobs \label{eq:vMILP_job_assignment} \\
\sum_{j\in\rjobs_{t}(p)}x_{j,t,C} &\leq C_{p}z_{C,t}   & \forall t\in [K],p\in \bjobsizes_t, C\in\Conf_t(T)  \label{eq:vMILP_big_jobs_bounded_by_conf}\\
\sum_{p\in\sjobsizes_t}p\sum_{j\in\rjobs_{t}(p)} x_{j,t,C} & \leq (T^D-\size(C)) z_{C,t} & \forall t\in [K], C\in\Conf_t(T) \label{eq:vMILP_jobs_fit_into_space}
\end{align}
Note that the last constraint is $D$-dimensional. 
Unlike in the other cases we can not transform an integral solution for $\MILP(T)$ directly into a schedule with only a small increase in the makespan.
However, we deal with this in the rounding step and still have:
\begin{lemma}\label{lem:vMILP_and_schedule}
If there is schedule with makespan $\makespan$ there is a feasible (integral) solution of $\MILP(\makespan)$. \qed
\end{lemma}
Using the algorithm by Lenstra and Kannan we can solve $\MILP(T)$ in time $f(1/\eps,D,K)\poly(|I|)$ for some computable function $f$.

\paragraph{Rounding.}

Using a variation of the rounding approach for the one dimensional case we can transform a solution $(z,x)$ for $\MILP(T)$ into a schedule with a makespan of at most $(1+\eps + \eps^2)T$.
The main difference is that we create nodes for pairs of \emph{machines} and processing time vectors instead of pairs of \emph{machine types} and processing times.

For each type $t$ we assign configurations to machines of type $t$ such that for each configuration $C \in \Conf_t(T) $ exactly $z_{C,t}$ configurations get assigned.
Therefore, we can assume that for each machine $i$ a configuration $C^{(i)}$ is given.
Based on this we can fractionally assign jobs to machines by setting $x_{j,i} = x_{j,t,C^{(i)}}/z_{C^{(i)},t}$, yielding:
\begin{equation}\label{eq:vfractional_schedule}
\sum_{j\in\jobs} p_{ij} x_{j,i} \leq T^D
\end{equation}
For each machine let $\vjobsizes_i$ be the set of occurring processing time vectors for machine $i$, that is, for each $p\in\vjobsizes$, we have $p\in\vjobsizes_i$, iff there is a job $j$ with $p_{ij} = p$ and $x_{j,i}>0$.
We set $\eta_{i,p}=\ceil{\sum_{j\in\rjobs_{t}(p)}x_{j,i}}$.
If $p$ is big, we have $\eta_{i,p}\leq C^{(i)}_p$, because of constraint~(\ref{eq:vMILP_big_jobs_bounded_by_conf}).

Like in the one dimensional case, the flow network $G=(V,E)$ has a source $\source$ and sink $\sink$, and for each job $j\in\jobs$ there is a job node $v_j$ and an edge $(\source,v_j)$ with capacity $1$ connecting the source and the job node.
Moreover, for each machine $i$ we have processing time vector nodes $u_{i,p}$ for each $p\in\vjobsizes_i$.
The processing time nodes are connected to the sink via edges $(u_{i,p},\sink)$ with capacity $ \eta_{i,p}$.
Lastly, for each job $j$ and machine type $i$ with $x_{j,i} > 0$, we have an edge $(v_j,u_{i,p_{i,j}})$ with capacity $1$ connecting the job node with the corresponding processing time vector nodes.
The variables $x_{j,i}$ yield a flow with value $n$ that is guaranteed to be correct because of the constraints of the MILP.
\begin{lemma}\label{lem:vmax_flow}
$G$ has a maximum flow with value $n$.\qed
\end{lemma}
Using the Ford-Fulkerson algorithm, an integral maximum flow $f^*$ can be found in time $\Oh(|E|f^*)=\Oh(n^2m)$.
Due to flow conservation, for each job $j$ there is exactly one machine $i^*$ such that $f((v_j,u_{i^*,p_{i^*j}}))=1$, and we set $\sigma(j) = i^*$.
Analogously to the one dimensional case, for each big processing time vector $p$ the schedule $\sigma$ assigns at most $C^{(i)}_p$ many jobs $j$ with $p_{ij}=p$ to machine $i$ and for each small processing time $p'$ vector one additional job $j$ with $p_{ij}=p'$ may be assigned to $i$.
Because of the choice of the parameter $\paramthresh$ and the second rounding step, we can bound the extra load $i$ receives.
\begin{lemma}
Let $q\in\jobsizes$ be small and $d\in [D]$. 
There are at most $(2\ceil{\log_{1+\eps}(D/\eps)})^{D-1}$ processing time vectors $p\in\vjobsizes$ with $p^{(d)}=q$.
\end{lemma}
\begin{proof}
Let $p$ be such a vector, $d'\neq d$ and $q' = p^{(d')}$.
Because of the second rounding step we have $q'\geq q\eps/D$ and $q \geq q'\eps/D $ $p^{(d')}$.
Now, because of the first rounding step there are only few such processing times $q$, more precisely at most $2\ceil{\log_{1+\eps}(D/\eps)}$.
Hence, there can be at most $(2\ceil{\log_{1+\eps}(D/\eps)})^{D-1}$ many such processing time vectors $p$.
\end{proof}
Using this lemma and the same argumentation as in the one dimensional case, we can bound the extra load machine $i$ receives in dimension $d$ by: 
\begin{align*}
& (2\ceil{\log_{1+\eps}(D/\eps)})^{D-1} \times  \paramthresh T \times \sum_{i=0}^{\infty} 1/(1+\eps)^i\\
\leq\ \  & 2(1/\eps \log((D-1)/\eps)+1)^{D-1}\times  (\eps^2/D)^D T \times (1+\eps)/\eps\\
\leq\ \  & 2(D/\eps^2)^{D-1} \times (\eps^2/D)^D T \times (1+\eps)/\eps\\
\leq\ \  & \eps^2 T \times (1+\eps)/\eps \leq (\eps + \eps^2)T
\end{align*}
Summarizing, we have:
\begin{lemma}
A solution $(z,x)$ for $\MILP(T)$ can be transformed into a schedule with makespan at most $(1+\eps+\eps^2)T$ in time polynomial in $|I|$ and $1/\eps$.
\end{lemma}
Therefore there is an EPTAS for this case as well.

\section{Conclusion}

We presented efficient approximation schemes for several variants of the problem of scheduling on unrelated parallel machines.
In the following, we briefly discuss some possible directions for further studies.

\paragraph{Better Running Times.}

The presented approximation schemes have running times of the form $f(1/\eps,K) + \poly(|I|)$ (or $f(1/\eps,K,D) + \poly(|I|)$ in the vector scheduling case).
While we took some effort to optimize $f$ at least for the first two schemes, we did not optimize the $\poly(|I|)$ part in any of the results.
Furthermore, for the case with a constant number of uniform types, one could study whether a quadratic or linear dependence in $1/\eps$ (ignoring polylogarithmic dependencies) in the exponent of the $f(1/\eps,K)$ part can be achieved, e.g. by utilizing techniques from \cite{Jan10QCmaxEPTAS} and \cite{JKV16ICALP}.
Lastly, the EPTAS for the vector scheduling variant is basically just a proof of concept and we did not optimize the running time at all.

\paragraph{Lower Bound.}

Chen, Ye and Zhang \cite{chen2014improved} showed that we can not hope for an EPTAS with a sub-linear dependency in $1/\eps$ in the exponent, unless the exponential time hypothesis fails.
It is unclear what can be ruled out in terms of the parameter $K$.

%
%

\paragraph{Job Types.}

In the introduction we mentioned the concept of job types for scheduling on unrelated parallel machines:
Two jobs $j,j'$ are of the same type, if they behave the same on every machine $i$, i.e., $p_{ij} = p_{ij'}$.
It is unknown, whether there is a PTAS for scheduling on unrelated parallel machines with a constant number of job types.
Furthermore, it is unknown, whether this problem is NP-hard.
Indeed, the problem is in P for important special cases:
For scheduling on identical parallel machines the number of job types is equal to the number of distinct processing times and for the case of the restricted assignment problem---where each job $j$ has a size $p_j$ and its processing time $p_{ij}$ on machine $i$ is either $p_j$ or $\infty$---the number of distinct processing times is bounded by the number of job types and a constant number of job types implies a constant number of machine types.
Both problems can be solved in polynomial time, if the number of distinct processing times is constant.

\paragraph{Acknowledgements.}

We thank Florian Mai and Jannis Mell for helpful discussions on the problem.

\bibliography{constant_mach_types.bib}

\end{document}